\documentclass{article}
\usepackage{amsthm, amsmath, amssymb, amsfonts}
\usepackage{mathtools}

\mathtoolsset{showonlyrefs}

\usepackage{dsfont}
\usepackage{enumitem}
\usepackage{booktabs}
\usepackage{tikz}
\usetikzlibrary{calc,arrows.meta}
\usepackage{pgf-pie}
\usepackage{comment}

\usepackage{subcaption} 

\usepackage{graphicx} 
\usepackage[dvipsnames]{xcolor}
\usepackage{url}

\usepackage{hyperref}

\newtheorem{theorem}{Theorem}[section]
\newtheorem{definition}[theorem]{Definition}
\newtheorem{proposition}[theorem]{Proposition}
\newtheorem{lemma}[theorem]{Lemma}

\theoremstyle{definition}
\newtheorem{example}[theorem]{Example}
\newtheorem{remark}[theorem]{Remark}

\newtheorem{assumption}[theorem]{Assumption}

\newtheorem{problem}[theorem]{Problem}

\newtheorem*{proposition*}{Proposition}
\newtheorem*{theorem*}{Theorem}
\newtheorem*{criterion*}{Criterion}
\newtheorem{remark*}{Remark}

\DeclareMathOperator{\Id}{Id} 
\DeclareMathOperator{\supp}{supp}

\def\bx{\mathbf{x}}
\def\bz{\mathbf{z}}
\def\by{\mathbf{y}}

\def\bq{\mathbf{q}}
\def\bP{\mathbf{P}}

\def\R{\mathbb{R}}
\def\cC{\mathcal{C}}
\def\cS{\mathcal{S}}

\newcommand{\eqdef}{\mathrel{\overset{\mathrm{def}}{=}}}

\usepackage{authblk}

\title{Non-conservative optimal transport}
\date{\today}

\author{
  Gabriela Kov\'a\v{c}ov\'a\,\thanks{
    Department of Mathematics, University of California Los Angeles, CA, USA
    \newline \hspace*{1.45em} Email: \url{kovacova@ucla.edu} and \url{gmenz@math.ucla.edu}
  \vspace{0.5em} },
  \quad
  Georg Menz\,\protect\footnotemark[1],  
  \quad 
  Niket Patel\,\protect\footnotemark[1] \hspace{1pt}\thanks{Center for Data Science, New York University, NY, USA \newline \hspace*{1.45em} Email: \url{nnp5656@nyu.edu}}
}

\begin{document}

\maketitle

\begin{abstract}

    Motivated by optimal re-balancing of a portfolio, we formalize an optimal transport problem in which the transported mass is scaled by a  mass-change factor depending on the source and destination. This allows direct modeling of the creation or destruction of mass. We discuss applications and position the framework alongside unbalanced, entropic, and unnormalized optimal transport. The existence of optimal transport plans and strong duality are established. The existence of optimal maps are deduced in two central regimes, i.e.~perturbative mass-change and quadratic mass-loss. For $\ell_p$ costs we derive the analogue of the Benamou–Brenier dynamic formulation.
    \end{abstract}

\section{Introduction}

Optimal transport originates with Monge’s 1781 sand-moving problem~\cite{monge1781memoire} and Kantorovich’s 1942 relaxation~\cite{kantorovich1942translation}, and has since blossomed via the Benamou–Brenier dynamic formulation~\cite{benamou2000computational} into the geometry of Wasserstein spaces. Over the last two decades, advances in duality, regularity, and scalable algorithms have made OT a central tool across analysis, probability, economics, imaging, and machine learning. For a thorough introduction to optimal transport we refer to the standard references~\cite{ambrosio2008gradient,villani2008optimal,santambrogio2015otam}. In the classical Kantorovich problem of optimal transport the goal is to find a coupling~$\pi$ of two given probability measures~$\mu$ and~$\nu$ that minimizes a given cost functional
\begin{align}
\inf_{\pi}  \ \int_{\mathcal{X} \times \mathcal{Y} } c(x,y) \pi(dx,dy) .   \tag{CKP}   \label{equ:classic_katorovich_problem}
\end{align}
Coupling means that the first marginal of~$\pi$ is given by~$\mu$ and the second marginal by~$\nu$, i.e.~for all test functions~$\xi$ and~$\zeta$
\begin{align}\label{equ:marginal_condition_CKP}
\hspace{-0.2cm}
    \int \xi(x) \pi(dx, dy) = \int \xi(x)\mu(dx) \quad \mbox{and} \quad  \int \zeta(y) \pi(dx, dy) = \int \xi(y)\nu(dy).
\end{align}

One of the limitations of the original optimal transport formulation is its rigidity in the marginal constraint~\eqref{equ:marginal_condition_CKP}. In particular, this constraint entails that the transport conserves the mass. Many applications require more flexibility and, therefore, variants of optimal transport allowing flexibility in the marginal constraints were developed and successfully applied in recent years, see Section~\ref{sec:unbalanced_optimal_transport} for a more detailed discussion including references. 
Building on this trajectory, we develop the framework of non-conservative optimal transport. \\

The framework originated when the authors studied a problem from financial mathematics, namely how to rebalance a portfolio of assets in an optimal way. In Section~\ref{sec:rebalance} we introduce the rebalancing problem and show existence of optimal rebalancing trade via tools of financial mathematics. However, studying this problem revealed a deeper structure. When exchanging one asset for another, a small percentage of its value is lost to trading costs (e.g.~bid-ask spread and trading fees). Mass is lost as the value of the portfolio after rebalancing is smaller than the value before. The shape of the target distribution is still fixed as one strives for obtaining the correct distribution of wealth across assets.\\

Motivated by this example, non-conservative optimal transport generalizes the classical optimal transport problem by two tweaks. First, one introduces a mass-change factor~$m(x,y)$ that models whether a mass transported from~$x$ to~$y$ is gained ($m(x,y) >1$), conserved ($m(x,y) =1$) or destroyed ($m(x,y) <1$).  Second, by weakening the second marginal condition such that the shape of the distribution of it is given by~$\nu$, i.e.
\begin{align}\label{equ:marginal_condition_KP} \hspace{-0.2cm}
    \int \xi(x) \pi(dx, dy) = \int \xi(x)\mu(dx) \quad \mbox{and} \quad  \frac{\int \zeta(y) m(x,y)\pi(dx, dy)}{\int m(x,y) \pi(dx,dy)} = \int \zeta(y)\nu(dy)
\end{align}
for all test functions~$\xi$ and~$\zeta$. Classical optimal transport is recovered  when $m \equiv 1$. \\

This framework allows for simple and direct modeling of transports that do not conserve mass. In Section~\ref{sec:non_conservative_optimal_transport} we discuss more examples like logistics with spoilage/yield, biochemical conversions, as well as economic scenarios like the influence of tariffs on the optimal distribution of goods, or the influence of retention rates onto the optimal distribution of jobs in an organization. Up to the knowledge of the authors, this simple framework to describe non-conservative optimal transport is new and not contained in the existing literature, though it can be described as an optimization problem over a family of unbalanced optimal transport problems. We refer to Section~\ref{sec:literature} for more details.\\

As the framework of non-conservative optimal transport stays close to the original Kantorovitch optimal transport problem, we are able to  study fundamental questions (like duality, existence of optimal maps, dynamic formulations) with similar methods, and come to satisfying answers. Summarizing, the article provides:
    \begin{itemize}
     \item A \emph{detailed motivating example} from financial mathematics (see~Section~\ref{sec:rebalancing_general_graphs}).
      \item The \emph{formulation} of \emph{non-conservative optimal transport} with cost \(c(x,y)\) and efficiency \(m(x,y)\), reducing to classical OT when \(m\equiv 1\) (see~Section~\ref{sec:nckp}).  
      \item The \emph{existence} of optimal plans (see Section~\ref{sec:nckp}).
      \item The \emph{dual problem} and \emph{strong duality} under mild continuity assumptions (see Section~\ref{sec:ncot_dual_problem} and~\ref{sec:strong_duality}).  
      \item The \emph{existence} of \emph{optimal maps} in perturbative (\(m \approx 1\)) and quadratic-cost regimes (see Section~\ref{sec:existence_optimal_transport_map}).
      \item A \emph{generalized Benamou--Brenier formula}, that is well-posed without \emph{a-priori} vector-field regularity (see Section~\ref{sec:dynamic_formulation}).  
      \item The relation to other optimal transport framworks allowing for mass creation or destruction (see Section~\ref{sec:literature}).  
    \end{itemize}
The regularity of optimal transport maps in the non-conservative framework is not discussed in this article and remains an open problem.

\section{Motivating example: Optimal rebalancing of portfolios.} \label{sec:rebalancing_general_graphs}
We open this paper with a motivating example from financial mathematics. Consider an investor whose investment strategy (or a target) is to holds $100 \cdot \nu_i \%$ of their portfolio value in asset $i$. If their currently held position does not correspond to this target proportion, then they need to buy and sell assets to achieve their target. We refer to this as rebalancing the investor's portfolio. Since trading incurs costs, we aim to find optimal rebalancing trade.\\

We do not address the question what is the best target portfolio, when to rebalance, or how far to jump in presence of fixed fees. Constructing optimal portfolios is a classical topic in financial mathematics and many different approaches are used. For an overview article on how optimal transport is used for portfolio construction we refer to~\cite{wong2019information}, but want to clearly separate our motivating example from it as it discusses a different aspect of portfolio theory. The question of when to rebalance leads to tracking problems, and it is known that including transaction costs leads to a no-trade region in which the marginal benefit of re-balancing further is smaller than the additional trading costs (see for example~\cite{davis1990portfolio,holden2013optimal} and references therein). Many of the tracking problems can be formulated as an entropic optimal transport problem, where the second marginal condition is relaxed to allow for a trade-off between hitting the target precisely vs.~the cost involved (see e.g.~Section~\ref{sec:entropic_optimal_transport}). \\

The rebalancing problem only becomes non-trivial in the presence of fees or transaction costs. It is discussed in~\cite{kabanov2009markets} in the case of separable trading costs (e.g.~markets with a single num\'eraire). This manuscript extends optimal rebalancing to the case of general markets with multiple num\'eraires. This is a non-trivial tasks as trading costs are not separable. \\

In Subsection \ref{sec:rebalance} we introduce the mathematical framework and state the optimal rebalancing problem. Then, in Subsection~\ref{sec:rebalanceOT} we interpret the problem in the framework of optimal transport. This will motivate us to introduce a notion of non-conservative optimal transport.

\subsection{Optimal rebalancing problem} \label{sec:rebalance}
The framework in which we state the problem is a general market model represented as a graph. We take a static, snap-shot view here.
\begin{definition}
A \emph{financial market} is modeled by a weighted directed graph~$\mathcal{G} = (V, E, P)$, where
\begin{itemize}
    \item vertices $V = \{1, \dots, N\}$ represent assets on the market,
    \item a directed edge $(i, j) \in E$ indicates that it is possible to directly trade asset $i$ into asset $j$,
    \item and the weight $P: E \to (0,\infty)$ represents the price quoted in units of asset received per unit of asset sold. For an edge $(i, j) \in E$, one unit of asset $i$ can be exchanged for $P_{(i,j)}$ units of asset $j$ on the market.
\end{itemize}
\end{definition}
For convenience, we extend the weight to $P: V\times V \to [0,\infty)$ by setting $P_{(i,i)} := 1$ and $P_{(i,j)} := 0$ for $(i,j) \not\in E$ with $i \neq j$. A market with single num\'eraire, referred through-out as \emph{star-shaped market} due to the shape of the corresponding graph, will be considered as a special case. However, we do not assume existence of a single num\'eraire in general. This allows to consider portfolios held in assets traded across different international markets and currencies or various crypto currencies. We refer to Figure~\ref{fig:general_financial_market} for an illustration.

\begin{assumption}
We make the following assumptions about the market.
\begin{itemize}
    \item[A1] \textit{Connectedness:}  $\mathcal{G}$ is a connected graph. 
    \item[A2] \textit{No arbitrage:} For every directed cycle $\Gamma$ in the graph $\mathcal{G}$ it holds
    $$\prod_{e \in \Gamma} P_e \leq 1.$$
    \item[A3] \textit{Unlimited liquidity and no price impact:} Assets are liquid and there is no limit on number of units that can be traded. Trading does not impact prices.
    \item[A4] \textit{No transaction costs:} There are no fees for trading.
\end{itemize}
\end{assumption}

\begin{remark} 
\begin{itemize}
    \item[(i)] The market model we introduce is static. We are taking a snap-shot view of considering the market at a single time point with the aim of implementing the rebalancing in a single trade executed immediately. This, alongside the assumption of no price impact, is a reasonable approximation when considering small portfolios. To consider large portfolios, price impact needs to be taken into account and a dynamic market model will be needed.
    \item[(ii)] We assume that all fees are incorporated in the quoted prices. It is possible to model linear (per share) fees as part of the bid-ask spread, but not flat fees.
    \item[(iii)] If we start with one unit of asset $i$ and trade along a cycle $\Gamma$, we hold $\prod_{e \in \Gamma} P_e$ units of asset $i$ after the trade. In particular, for any pair of directly tradable assets $i$ and $j$ no arbitrage implies
    $$P_{(i,j)}  \leq \frac{1}{P_{(j,i)}}.$$
\end{itemize}
\end{remark}

\begin{example}\label{ex:star1}
Consider a market with a single num\'eraire (w.l.o.g.~asset $1$), where assets $=2, \dots, N$ can only be traded for the num\'eraire. We refer to this as a \emph{star-shaped market}. It is usual to quote prices in units of the num\'eraire as the bid and ask prices,
\[
P^b_i  = P_{(i, 1)} \quad \text{ and } \quad P^a_i = \frac{1}{P_{(1, i)}} \quad \text{ for assets } i = 2, \dots, N.
\] 
For completeness we set $P^b_1 = P^a_1 := 1$.
\end{example}

\begin{figure}
    \centering
    \begin{tikzpicture}[scale =0.8, >=stealth,
    hub/.style   ={circle,draw,thick,minimum size=14mm,font=\small},
    asset/.style ={circle,draw,thick,minimum size=4.5mm,inner sep=0pt},
    enum/.style  ={circle,draw,thick,fill=white,font=\tiny,inner sep=0pt,minimum size=3.6mm},
    edge/.style  ={thick}]

\coordinate (USD) at (-4, 0);
\coordinate (EUR) at ( 4, 0);
\coordinate (JPY) at ( 0,-4);

\coordinate (TOP)  at ( 0, 1.5);
\coordinate (MID)  at ( 0,-1.5);
\coordinate (LLOW) at (-3,-3);
\coordinate (RLOW) at ( 3,-3);

\foreach \name/\ang in {u1/135, u2/180,u3/225}
  \coordinate (\name) at ($(USD)+({1.6*cos(\ang)},{1.6*sin(\ang)})$);

\foreach \name/\ang in {e1/ 45,e2/  0,e3/-45}
  \coordinate (\name) at ($(EUR)+({1.6*cos(\ang)},{1.6*sin(\ang)})$);

\foreach \name/\ang in {j1/-150,j2/-90,j3/-30}
  \coordinate (\name) at ($(JPY)+({1.6*cos(\ang)},{1.6*sin(\ang)})$);

\node[hub] (U) at (USD) {USD};
\node[hub] (E) at (EUR) {EUR};
\node[hub] (J) at (JPY) {JPY};

\foreach \n in {TOP,MID,LLOW,RLOW,u1,u2,u3,e1,e2,e3,j1,j2,j3}
  \node[asset] (\n) at (\n) {};

\draw[edge] (U)--(J);
\draw[edge] (U)--(E);
\draw[edge] (E)--(J);
\draw[edge] (U)--(LLOW);
\draw[edge] (J)--(LLOW);
\draw[edge] (E)--(RLOW);
\draw[edge] (J)--(RLOW);
\draw[edge] (U)--(TOP);
\draw[edge] (E)--(TOP);

\foreach \cross in {MID}
 \foreach \hub in {U,E,J} \draw[edge] (\hub)--(\cross);

\foreach \p in {u1,u2,u3} \draw[edge] (U)--(\p);
\foreach \p in {e1,e2,e3} \draw[edge] (E)--(\p);
\foreach \p in {j1,j2,j3} \draw[edge] (J)--(\p);

\end{tikzpicture}

    \caption{Illustration of a general financial market without a single num\'eraire. The vertices USD, JPY and EUR denote currencies used a num\'eraires at national markets. Vertices connected to several num\'eraires denote cross-traded assets, i.e.~assets that are traded on several national markets.}
    \label{fig:general_financial_market}
\end{figure}
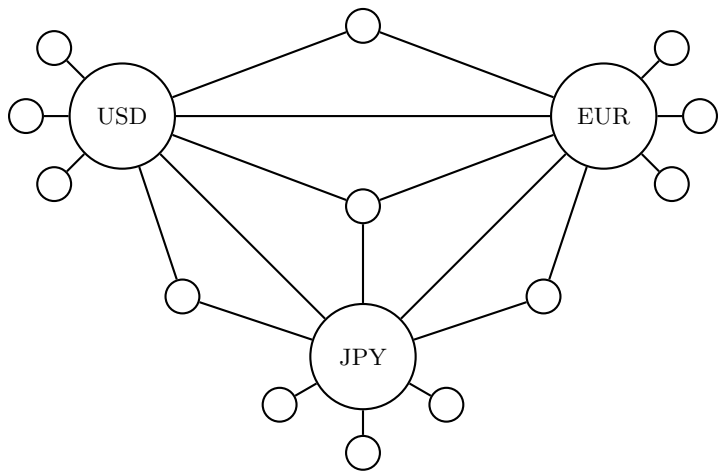

\begin{definition}
A \emph{portfolio} is given by a vector $\bx = (x_1, \dots, x_N)$, where $x_i$ denotes the number of units of the asset $i$ held in the portfolio.

An \emph{admissible trade} is a map $\xi: V \times V \to [0, \infty)$ satisfying $\xi_{(i, j)} = 0$ for $(i,j) \not\in E$. Here $\xi_{(i, j)}$ denotes the number of units of asset $i$ that are (directly) traded for asset $j$ on the market. $\Xi$ denotes the set of all admissible trades.
\end{definition}
An admissible trade $\xi$ changes the number of units of asset $i$ held by
\begin{align}
\label{eq:Deltax}
\Delta x_i (\xi) = \sum_{j=1}^N P_{(j, i)} \xi_{(j, i)} - \xi_{(i,j)}, \quad \quad i = 1, \dots, N.    
\end{align}

\begin{example}\label{ex:star2}
Consider an admissible trade on the \emph{star-shaped market}. A positive bid-ask spread ensures that a round-trip trade (buying and selling the same asset) results in a net loss, effectively 'burning money'. By imposing 
\begin{align}
\label{eq:star1}
    \xi_{(1,i)} \cdot \xi_{(i,1)} = 0 \quad \quad  i = 2, \cdots, N
\end{align}
we exclude round-trip trades. Then, a trade can be represented by change in the portfolio position $\Delta \bx$. Since under \eqref{eq:star1} it holds $(\Delta x_i)^+ = \frac{1}{P^a_i} \xi_{(1,i)}$ and $(\Delta x_i)^- = \xi_{(i,1)}$ for $i > 1$, \eqref{eq:Deltax} implies the \emph{self-financing condition}
\begin{align*}
\Delta x_1  = \sum_{i=2}^N P^b_i (\Delta x_i)^- - P^a_i (\Delta x_i)^+.  
\end{align*}
\end{example}

The aim is to rebalance the current portfolio $\bx$ in such a way that the \emph{target proportions} $\nu = (\nu_1, \dots, \nu_N)$ are achieved. That is, we want $100 \cdot \nu_i \%$ of portfolio value is held in asset $i$ after the rebalancing trade.\\

In order to measure (percentage of) portfolio value, the investor needs to decide on what they believe to be the (`true' or `fair') value of each asset. We assume that these values are consistent with the prices quoted on the market. Therefore, implicitly we are assuming the market prices are not misquoted. 
\begin{definition}
A relative price vector $\bq: V \to (0, \infty)$ is \emph{consistent} with the financial market $\mathcal{G} = (V, E, P)$ if it holds
\[
P_{(i,j)} \leq \frac{q_i}{q_j} \quad \quad \text{ for all } (i, j) \in E.
\]
\end{definition}

\begin{remark} 
\begin{itemize}
    \item[(i)] Consistent price vector is scale invariant\footnote{If $\bq$ is a consistent price vector, then $\lambda \cdot \bq$ is consistent price vector for all $\lambda > 0$.}. We can, therefore, turn any asset $i$ into a \emph{hypothetical num\'eraire} by imposing $q_i = 1$. In the sequel, asset $1$ will be designated as the hypothetical num\'eraire. 

    \item[(ii)] Consistent price vector $\bq$ defines a consistent pricing measure through $Q_{(i,j)} := \frac{q_i}{q_j}$ for all $i, j \in V$. \emph{Consistent pricing measure} is $Q: V \times V \to [0,\infty)$ satisfying
    \begin{itemize}
    \item[(p1)] no bid-ask spread: $Q_{(i,j)} = \frac{1}{Q_{(j, i)}}$  for all $(i,j)\in E$, 
    \item[(p2)] consistency: $P_{(i,j)} \leq Q_{(i,j)} \leq \frac{1}{P_{(j, i)}}$  for all $(i,j)\in E$,
    \item[(p3)] path independence: for all $i, j \in V$ and all directed walks $\Gamma_1, \Gamma_2$ connecting vertex $i$ to vertex $j$ it holds 
    \[
     \prod_{e \in \Gamma_1} Q_e = \prod_{e \in \Gamma_2} Q_e.   
    \]
    \end{itemize}
    \item[(iii)] Conversely, existence of a consistent pricing measure satisfying properties (p1)-(p3) implies existence of a consistence price vector.

    \item[(iv)] The decision maker on a market with a bid-ask spread could instead of a single (point) relative price of two asset consider a distribution over an interval to be a fair relative price.
\end{itemize}    
\end{remark}

\begin{example}\label{ex:star3}
Once more, consider the \emph{star-shaped market}, where the no arbitrage condition (A2) simplifies to
\[
P^b_i \leq P^a_i \quad \quad \quad  i = 2, \dots N.
\]
Any choice of $P^b_i \leq q_i \leq P^a_i$ leads to a consistent price vector, one such example is the mid price. Observe that in presence of a non-trivial bid-ask spread, the consistent price vector is not unique (despite fixed num\'eraire).
\end{example}

Before addressing the rebalancing problem, we note that no arbitrage is both necessary and sufficient for existence of consistent price vector. Proof of Proposition~\ref{prop:consistent-price} is given in Appendix~\ref{append:consistent-price}.
\begin{proposition}\label{prop:consistent-price}
    Let the financial market $\mathcal{G} = (V, E, P)$ satisfy Assumption A1 (connectedness). A consistent price vector exists if and only if the the market satisfies Assumption A2 of no arbitrage.
\end{proposition}

\begin{remark}
    The Bellman-Ford algorithm, see~\cite[Section 24.1]{CormenETAL09}, can be used to check whether the market $\mathcal{G}$ satisfies the no arbitrage condition and a consistent price vector exists.
\end{remark}

\begin{assumption}
In addition to A1-A4 we assume the following.
\begin{itemize}
    \item[A5] \textit{Long only portfolio:} Current portfolio $\bx$ consists only of long positions, $x_i \geq 0$ for $i=1, \dots, N$.
    \item[A6] \textit{Target proportions:} A vector $\nu$ satisfying $\sum_{i=1}^N \nu_i = 1$ and $\nu_i \geq 0$ for $i=1, \dots, N$ is given.
    \item[A7] \textit{Consistent price vector:} A consistent price vector $\bq$ is fixed. Without loss of generality, we impose $q_1 =1$ and treat asset $1$ as the hypothetical num\'eraire.
\end{itemize}
\end{assumption}

In what follows we use the fixed consistent price vector $\bq$ to compute value of  a portfolio expressed in units of the hypothetical num\'eraire. 
Consider the current portfolio $\bx$ and an admissible trade $\xi$. The value held in asset $i$ before and after the trade, respectively,  are
\begin{align*}
    q_i x_i \quad \text{ and }  \quad q_i \left( x_i +   \Delta x_i (\xi) \right). 
\end{align*}

Note that while we use the consistent prices $\bq$ to evaluate positions, trades are executed at the market prices $P$. Therefore, a trade 
decreases the value of the portfolio.  
We denote the \emph{cost of (rebalancing) trade $\xi$} by 
\begin{align*}
    C(\xi) &:=  - \sum_{i=1}^N q_i \Delta x_i(\xi).
\end{align*}

We aim to find a trade $\xi$ such that the value of portfolio $\bx + \Delta \bx (\xi)$ is distributed according to the target proportions $\nu$ and to do so at the lowest cost possible. For an illustration of the reabalancing problem we refer to Figure~\ref{fig:rebalancing}. This yields the following optimization problem.

\begin{problem}\label{prob:rebalance}
We identify the optimal rebalancing trade as the solution of 
\begin{align*}
\underset{\xi \in \Xi}{\text{minimize }} \quad  &C(\xi) & \\
\text{subject to } \quad &x_i + \Delta x_i (\xi) \geq 0, \quad \quad &i = 1, \dots, N,  \\
&\nu_i = \dfrac{q_i \left( x_i + \Delta x_i (\xi) \right)}{\sum\limits_{k=1}^N q_k \left( x_k + \Delta x_k (\xi) \right)} \quad \quad &i = 1, \dots, N. 
\end{align*}
We denote the feasible set of this problem by $\Xi(\nu)$.
\end{problem}

\begin{remark}
Since it holds
\begin{align*}
\sum\limits_{k=1}^N q_k \left( x_k + \Delta x_k (\xi) \right) = \sum\limits_{k=1}^N q_k  x_k - C(\xi), 
\end{align*}
minimizing the rebalancing cost is equivalent to maximizing the post-rebalancing value of the portfolio,
\begin{align*}
\underset{\xi \in \Xi(\nu)}{\text{arg min }} C(\xi) &= \underset{\xi \in \Xi(\nu)}{\text{arg max }} \sum\limits_{k=1}^N q_k \left( x_k + \Delta x_k (\xi) \right). 
\end{align*}
\end{remark}

\begin{theorem}\label{thm:OptRebalancing}
Let the assumptions A1-A7 hold. Solution to the optimal rebalancing Problem~\ref{prob:rebalance} exists.    
\end{theorem}
In the next subsection, we show that the rebalancing problem corresponds to a \emph{non-conservative optimal transport}. Theorem~\ref{thm:OptRebalancing} will, therefore, follow from existence of an optimal solution of a non-conservative OT, which will be discussed in Section~\ref{sec:non_conservative_optimal_transport}, see Theorem~\ref{the:solution_KP}. An alternative, direct proof of Theorem~\ref{thm:OptRebalancing} is provided in Appendix~\ref{append:existenceRebalancing} through a series of partial results.

\begin{figure}
    \centering

    \begin{tikzpicture}[scale=0.75, transform shape]
  \node (leftPie) at (0,0) {
    \begin{tikzpicture}
      \pie[
        text=inside,
        radius=2,         
        sum=100,
        color={red!70,blue!70,green!70,yellow!70}
      ]{
        30/GLD,
        30/BTC,
        20/SPY,
        20/QQQ
      }
    \end{tikzpicture}
  };

  \node (rightPie) at (8,0) {
    \begin{tikzpicture}
      \pie[
        text=inside,
        radius=2,         
        sum=100,
        color={red!70,blue!70,green!70,yellow!70}
      ]{
        25/GLD,
        25/BTC,
        25/SPY,
        25/QQQ
      }
    \end{tikzpicture}
    
  };

  \draw[
    ->, thick,
    shorten >=2mm,    
    shorten <=2mm     
  ]
    (leftPie.east) -- node[above, font=\bfseries]{REBALANCING} (rightPie.west);

  \node[anchor=north, yshift=-2mm] at (leftPie.south)  {Source measure \(\mu\)};
  \node[anchor=north, yshift=-2mm] at (rightPie.south) {Target measure \(\nu\)};
\end{tikzpicture}

    \caption{Illustration of the rebalancing process. The goal is to sell and buy assets such that the wealth of the new portfolio is distributed according to the target proportions (or measure)~$\nu$. }
    \label{fig:rebalancing}
\end{figure}
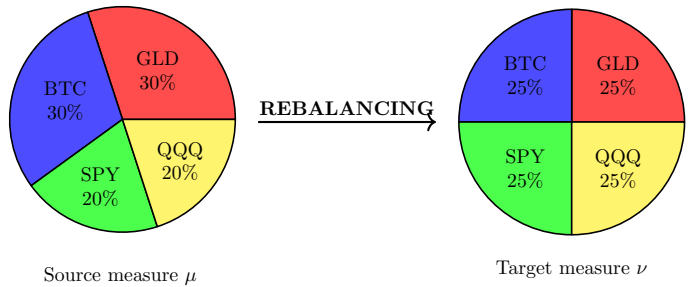

\begin{remark}[Extensions of optimal re-balancing]
Solving the optimal re-balancing Problem~\ref{prob:rebalance} is only a first (theoretical) step. When trading larger quantities, price impact cannot be ignored and leads to a non-linear optimization problem that is either convex or nonconvex depending on how price impact is modeled. 

As known from optimal execution problems, re-balancing a portfolio all at once also might not be optimal as it might be beneficial to spread out trades over time to reduce the price impact. This would require the modeling of dynamic correlated asset prices and price impacts. Optimal rebalancing can then be approached via the dynamic formulation of optimal transport (cf.~Section~\ref{sec:dynamic_formulation}). Here, the portfolio is rebalanced over one unit of time~$0 \leq t \leq 1$ and the optimal transport is encoded via a vector field~$v(t,x)$ describing the infinitesimal rate of trading at time~$t$. Another important question is how often and when one should re-balance (see for example~\cite{morton1995optimal,holden2013optimal}), leading to a tracking problem, where the risk of holding the wrong portfolio needs to be quantified.
\end{remark}

\subsection{Rebalancing and optimal transport} \label{sec:rebalanceOT}

Let us now show that the optimal portfolio rebalancing problem is closely connected to optimal transport (OT) on a discrete probability space $\Omega = V = \{1, \dots, N\}$. In rebalancing, we transport outstanding distribution of portfolio value into target distribution while minimizing trading costs. The main difference to the classical OT framework is that trades decrease overall portfolio value while in classical OT the overall mass is conserved. This will motivate a transport problem where mass is lost during transport. We call this a \emph{non-conservative optimal transport} problem.\\

Given the current portfolio position $\bx$ and investor's consistent price vector $\bq$, the proportion of current portfolio value held in asset $i$ is
\[
\mu_i = \frac{q_i \cdot x_i}{\sum_{k=1}^N q_k \cdot x_k} \quad \quad \text{ for } i = 1, \dots, N.
\]
The current proportions $\mu = (\mu_1, \dots, \mu_N)$ as well as the target proportions $\nu = (\nu_1, \dots, \nu_N)$ are probability distributions on $\Omega$. We can, therefore, consider transport plans $\pi: \Omega \times \Omega \to [0,1]$ between measures $\mu$ and $\nu$. The value $\pi(i,j)$ should represent \emph{proportion of portfolio's value} that should be transported (traded) from asset $i$ to asset $j$. The cost functional $c: \Omega \times \Omega \to [0, \infty]$ capturing the proportional loss of portfolio value when exchanging (trading) assets is given by
\begin{align}
    c(i,j) = \begin{cases}
    0, & \text{ if } i = j, \\
    1 - \frac{q_j}{q_i} P_{(i,j)}, & \text{ if } (i, j) \in E, \\
    \infty, & \text{ else}.
    \end{cases}
\end{align}

\begin{remark}
Let us relate transport plan $\pi$ and cost $c$ to trade $\xi$ and its cost $C(\xi)$. In view of~\eqref{eq:Deltax}, the cost of a trade equals
\begin{align} \label{eq:costC}
\begin{split}
C(\xi) &= \sum\limits_{i=1}^N q_i \sum\limits_{j=1}^N  (\xi_{(i, j)} - P_{(j,i)} \xi_{(j,i)}) = \sum\limits_{(i,j) \in E} (q_i - q_j P_{(i,j)}) \xi_{(i,j)}.
\end{split}
\end{align} 
Since the transport plan $\pi$ replace trade $\xi$ as a variable, (optimal) $\pi(i,j)$ should be proportional to (optimal) $q_i \xi_{(i,j)}$. Consequently, for $(i,j) \in E$ the cost associated with $\pi(i,j)$ shall be $c(i,j) = \frac{1}{q_i} (q_i - q_j P_{(i,j)})$.  Holding (not trading) asset $i$ does not incure cost, so $c(i,i) = 0$ and by imposing infinite cost $c(i,j) = \infty$ on a pair of assets that cannot be directly traded we enforce that optimal $\pi$ shall correspond to an admissible trade.
\end{remark}

Let us consider the (classical) optimal transport problem with cost $c$.
\begin{problem} \label{prob:Rebalance-OT1}
We consider the problem 
\begin{align*}
    \min_{\pi \in \Gamma(\mu, \nu)} \sum_{i,j =1}^N c(i,j) \cdot  \pi(i,j),
\end{align*}
 where 
\begin{align*}
    \Gamma(\mu, \nu) : = \left\{ \pi: \Omega \times \Omega \to [0,1] \ : \ \sum_{j=1}^N \pi(i,j) = \mu_i \quad \mbox{and} \quad \sum_{i=1}^N \pi(i,j) = \nu_j  \right\}.
\end{align*}
\end{problem}

It is known that an optimal solution $\tilde{\pi}$ of the optimal transport Problem~\ref{prob:Rebalance-OT1} exists. Using the optimal transport plan $\tilde{\pi}$, we can construct a trade $\xi^{\tilde{\pi}}$ as
\[
\xi^{\tilde{\pi}}_{(i,j)} := \frac{\sum_{k=1}^N q_k x_k}{q_i} \cdot \tilde{\pi}(i, j) \quad \quad \text{ for } i \neq j.
\]
In light of graph $\mathcal{G}$ being connected (assumption A1), problem admits a feasible solution with a finite objective value. Therefore, the optimal transport plan $\tilde{\pi}$ satisfies $\tilde{\pi}(i, j) = 0$ for $(i,j) \not\in E$ with $i \neq j$ and, consequently, $\xi^{\tilde{\pi}} \in \Xi$ is an admissible trade.

However, one can show that for an asset $i$ it holds
\[
x_i + \Delta x_i (\xi^{\tilde{\pi}}) = \frac{\sum_{k=1}^N q_k x_k}{q_i} \cdot \sum_{k=1}^N \frac{q_k}{q_i} P_{(i,k)} \tilde{\pi} (k, i).
\]
Therefore, in presence of a non-trivial bid-ask spread, $\bx + \Delta \bx(\xi^{\tilde{\pi}})$ does not lead to the desired target proportions $\nu$ and $\xi^{\tilde{\pi}}$ is not a feasible rebalancing trade. This is due to $\Gamma(\mu, \nu)$ not accounting for loss in value resulting from transport (trade) from asset $i$ to asset $j$. 

To formulate a transport problem corresponding to the optimal portfolio rebalancing, we need to account for losses of mass during the transport. For this purpose, we introduce mass-change function $m: \Omega \times \Omega \to [0, \infty]$ given by 
\[
m (i,j) := \frac{q_j}{q_i} P_{(i,j)}.
\]
\begin{remark}
On the market, one unit of asset $i$ (with value $q_i$) can be traded for $P_{(i,j)}$ units of asset $j$ (with value $q_j P_{(i,j)}$). Therefore, when trading along an edge $(i,j) \in E$, a proportion $m (i,j) := \frac{q_j}{q_i} P_{(i,j)}$ of value is preserved.
\end{remark}

\begin{problem} \label{prob:Rebalance-OT2}
Now, we consider problem
\begin{align*}
    \min_{\pi \in \Gamma_m (\mu, \nu)} \sum_{i,j =1}^N c(i,j) \cdot  \pi(i,j),
\end{align*}
 where 
\begin{align*}
    \Gamma_m (\mu, \nu) : = \left\{ \pi: \Omega \times \Omega \to [0,1] \ : \ \sum_{j=1}^N \pi(i,j) = \mu_i \quad \mbox{and} \quad \dfrac{\sum\limits_{i=1}^N m(i,j)\pi(i,j)}{\sum\limits_{k=1}^N\sum\limits_{i=1}^N m(i,k)\pi(i,k)} = \nu_j  \right\}.
\end{align*}
\end{problem}

Problem~\ref{prob:Rebalance-OT2} is a modification of optimal transport problem with mass being lost during the transport, that is a non-conservative optimal transport. Section~\ref{sec:non_conservative_optimal_transport} studies the general case of a non-conservative optimal transport problem. Among other results, existence of an optimal solution will be proven in Theorem~\ref{the:solution_KP}. To conclude this subsection, we illustrate that the optimal rebalancing Problem~\ref{prob:rebalance} and the modified transport Problem~\ref{prob:Rebalance-OT2} are equivalent.

\begin{proposition} \label{prop:rebalance2OT}
Let Assumptions A1-A7 hold and denote $v:= \sum_{k=1}^N q_k x_k$. Trade $\xi \in \Xi(\nu)$ is an optimal solution of Problem~\ref{prob:rebalance} if and only if $\pi^\xi$ given by
\begin{align*}
\pi^\xi (i,j) = \begin{cases}
    \frac{q_i}{v} \xi_{(i,j)} & \text{ if } i \neq j, \\
    \frac{q_i}{v} \left( x_i - \sum_{k \neq i} \xi_{(i,k)} \right) & \text{ if } i = j,
\end{cases}    
\end{align*}
is an optimal solution of Problem~\ref{prob:Rebalance-OT2}. 
\end{proposition}
Proof of Proposition~\ref{prop:rebalance2OT} is given in Appendix~\ref{append:rebalance2OT} following a partial result.

\section{Non-conservative optimal transport}\label{sec:non_conservative_optimal_transport}

This section develops the framework and theory of non-conservative optimal transport. In Section~\ref{sec:nckp} we propose and solve the non-conservative Kantorovich problem and discuss applications. In Section~\ref{sec:ncot_dual_problem}, we develop and solve the dual-problem. Strong duality is deduced in Section~\ref{sec:strong_duality}. The existence of optimal transport maps is discussed in Section~\ref{sec:existence_optimal_transport_map} and the dynamic formulation is treated in Section~\ref{sec:dynamic_formulation}.

\subsection{The non-conservative Kantorovich Problem.}\label{sec:nckp}

Let us consider two complete, separable metric spaces~$\mathcal{X}$ and~$\mathcal{Y}$.  The mass-change factor~$m: \mathcal{X} \times \mathcal{Y} \to (0, \infty)$ models the percentage of mass that is lost (if $m(x,y) \leq 1$) or gained (if $m(x,y) >1$) when transporting some mass from~$x$ to $y$. The loss or gain of mass is \textit{proportional} to the amount of transported mass~$\pi(x,y)$. The mass-change factor~$m(x,y)$  depends only  on the starting point~$x$ and endpoint~$y$, and not on other factors like the amount of mass transported, the route or the velocity of the transport. 

We consider a source probability measure~$\mu \in \mathcal{P}(\mathcal{X})$ and transport its mass according to a transport plan~$\pi \in \mathcal{P(} \mathcal{X} \times \mathcal{Y})$. As mass is lost or gained during the transport, the transported measure is distributed according to
\begin{align*}
  \mathcal{B}(\mathcal{Y}) \ni B \mapsto  \int_{\mathcal{X} \times B} m(x,y) \pi(dx,dy),
\end{align*}
which generically is not a probability measure. However, we require that the remaining mass after transport will be distributed \textit{proportionally} to a given target probability measure~$\nu \in \mathcal{P}(\mathcal{Y})$. Hence, the set of non-conservative transport plans w.r.t.~source measure~$\mu$ and target measure~$\nu$ is given by
\begin{align}
\begin{split}\label{equ:def_gamma_m}
\Gamma_m(\mu, \nu)  = \Big\{  & \pi \in \mathcal{P} (\mathcal{X} \times \mathcal {Y} ) \ :    \  \forall A \in \mathcal{B}(\mathcal{X}) \int_{ A \times \mathcal{Y}} \pi(dx,dy) = \mu(A),   \\ 
& \quad \quad \quad   \ \forall B \in \mathcal{B}(\mathcal{Y}) \  \frac{\int_{\mathcal{X} \times B} m(x,y) \pi(dx,dy)}{\int_{\mathcal{X} \times \mathcal{Y}} m(x,y)\pi(dx,dy)} =  \nu(B)  \Big\}. 
\end{split}
\end{align}
This means that the first marginal is distributed as~$\mu$, i.e.~$(\mbox{Proj}_1)_{\sharp} \pi = \mu$. The second marginal satisfies the relation~$(\mbox{Proj}_2)_{\sharp} (m \pi) = \left( \int m \ d \pi  \right)  \cdot \nu$.

\begin{definition}[Kantorovich problem in non-conservative optimal transport] \label{lossyOTplans}
 Given a cost function~$c: \mathcal X \times \mathcal{Y} \to \mathbb{R}$ and  a mass-change factor~$m: \mathcal{X} \times \mathcal{Y} \to \mathbb{R}$, optimize the expression
\begin{align}
\inf_{\pi \in \Gamma_m(\mu, \nu)}  \ \int_{\mathcal{X} \times \mathcal{Y} } c(x,y) \pi(dx,dy) .   \tag{KP}   \label{equ:katorovich_problem}
\end{align}
\end{definition}
\begin{remark}[Feasibility of \eqref{equ:katorovich_problem}]
    Under general non-degeneracy assumptions the set~$\Gamma_{m}(\mu, \nu)$ is non-empty. More precisely, if the measure~$\tilde \nu$ given by 
    \begin{align*}
        \tilde \nu(B) := \frac{1}{Z}\int_{B} \frac{1}{\int_{\mathcal{X}}  m(x,y) \mu(dx)} \nu(dy) \quad \mbox{with~$0<Z<\infty$ such that }  \tilde\nu (\mathcal{Y}) =1
    \end{align*}
    is well-defined, it holds~$\mu \otimes \tilde \nu \in \Gamma_{m} (\mu, \nu)$.
\end{remark}
We assume for the rest of this article that~$\Gamma_m(\mu, \nu)$ is non-empty. Similar to standard optimal problem, the non-conservative Kantorovitch problem admits a solution under general assumptions.
\begin{theorem}[Solution of~\eqref{equ:katorovich_problem}]\label{the:solution_KP}
	Suppose that the cost function~$c$ is lower semi-continuous and bounded from below, and that the mass-change factor~$m$ is continuous and bounded. Then the Problem \ref{lossyOTplans} admits a minimizer~$\pi^* \in \Gamma_m(\mu, \nu) $, i.e.
    \begin{align}
        \inf_{\pi \in \Gamma_m(\mu, \nu)}  \ \int_{\mathcal{X} \times \mathcal{Y} } c(x,y) \pi(dx,dy)  =   \int_{\mathcal{X} \times \mathcal{Y} } c(x,y) \pi^*(dx,dy).
    \end{align}
\end{theorem}
The proof of Theorem~\ref{the:solution_KP} follows the standard argument for the traditional Kantorovich problem in optimal transport, (see e.g.~proof of Theorem 1.7 in~\cite{santambrogio2015otam}). The assumption of the mass-change factor being bounded and continuous is used to show that the subsequential limit~$\pi^*$ of the minimizing sequence is in the admissible set~$\Gamma_m(\mu, \nu)$ by the definition of weak-convergence. Lower-semicontitnuity yields that~$\pi^*$ is optimal.

\begin{remark}[Non-uniqueness of the optimal transport plan]
    In general, the optimal transport plan~$\pi^*$ is not unique. Let us consider a mass-change factor such that~$m(x,y)=0$ whenever~$|x-y| \geq C >0$ (see e.g.~\eqref{equ:loss_function_leaky_monge} below).  If the source measure~$\mu$ that has some mass in a set~$A$ such that~$\inf_{x \in A, y \in \supp \nu} |x-y| \geq C$, the mass from~$A$ cannot reach the support~$\supp \nu$ and it does not matter where in~$\supp \nu$ it gets transported to. Remark~\ref{rem:uniqueness_via_optimal_transport_maps} in Section~\ref{sec:existence_optimal_transport_map} discusses a situation where the optimal transport plan is unique. 
\end{remark}

In the remainder of this subsection we present settings and situations where the non-conservative OT may arise or find applications. A detailed example of optimal rebalancing of portfolios corresponding to a discrete non-conservative optimal transport problem was outlined in Section~\ref{sec:rebalancing_general_graphs}.
\begin{example}[Leaky Monge Problem]
    The original motivation of Monge was to find the most efficient way to move sand distributed in space acording to~$\mu$ to a desired target distribution~$\nu$ using buckets of unit mass. Therefore, the cost function~$c(x,y)$ is given by the natural (e.g.~Euclidean or Manhattan) distance between~$x$ and~$y$. 
    The solution of the associated optimal transport problem minimizes the mean dislocation. No particle of sand is making an unnecessary round trip. If the mass is only moved horizontally on a frictional surface with constant resistive force proportional to its weight, this corresponds to minimizing the needed work to transport the sand. 
    
    Now, let us assume that buckets are leaky, i.e.~they loose sand at a uniform rate due to a small hole in the bucket. Assuming that mass is transported with uniform speed, the mass-change factor is given by
    \begin{align}\label{equ:loss_function_leaky_monge}
        m(x,y) = \max\{0, 1- k d(x,y)\},
    \end{align}
    where $d(x,y)$ measures the time to transport and the constant~$k>0$ models the size of the hole.
    
    In many physical models, the loss of mass would not happen at a uniform rate but be proportional to the remaining sand in the bucket. Examples would be loss of mass due to radioactive decay, first order chemical reactions such as drug elimination in a host, small gas leaks, or biodegradation of organic matter. In such a situation the loss will be exponential in the time needed to transport the mass from~$x$ to~$y$ and the mass-change factor is given by
    \begin{align}\label{equ:loss_function_proportional_leaky_monge}
        m(x,y) = \exp(- k   d(x,y)).
    \end{align}
    A situation when mass is gained in the transport could be a vehicle that charges its battery via solar panels or an infrastructure (e.g.~satellites or cars). 
\end{example}

\begin{remark}[Limitations and extensions]
    Our model cannot account for situations when the loss of mass is not proportional to the amount of transported mass. In reality, the dependency between cost function, the mass-change factor, and the transport map is a lot more subtle and should be carefully modeled with the specific application in mind. The interplay between those quantities gives rise to richer models and strategies.  In the application to optimal rebalancing, the cost function is equivalent to minimize the loss of mass. Non-conservative optimal transport problems are able to find find the balance between competing goals; for example when finding the best speed to transport mass from~$x$ to~$y$ in order not to loose too much mass but also not to spend too much energy for the transport.
\end{remark}

Optimal transport has beautiful interpretations and applications in economics, see e.g.~\cite[Section~1.7.3]{santambrogio2015otam}. We mention two such problems from economics here for which a non-conservative extension has a natural applications. 

\begin{example}[Utility maximization]
     One such application is in optimal allocation of goods to consumers. The source measure~$\mu(x)$ and the target measure~$\nu(y)$ represent the distributions of various types of goods and consumers, respectively. The aim to the optimal transport problem is to maximize the overall utility $u = -c$ across possible matchings $\pi$ of goods to consumers. The optimal potential~$\varphi (x)$ of the dual problem (cf.~Section~\ref{sec:ncot_dual_problem} below) determines the optimal prices of the good~$x$, and the potential~$\psi(y)$ describes the received utility of the consumer~$y$. One remarkable feature is that in the optimal state, every single utility is maximized i.e.~$\psi(y) = \max_{x} (u_x - \varphi(x))$. 
    
    In a non-conservative setting, the mass-change function~$m(x,y)$ can describe the taxes or tariffs customer~$y$ has to pay when buying the good~$x$. Non-conservative optimal transport provides a tool to study the implications of new tariffs. It could also give rise to new tools that allows to find the optimal design of tariffs given a certain agenda.
\end{example}

\begin{example}[Productivity maximization]    
    Another economic application is the maximization of the productivity of a company. The source measure~$\mu(x)$ is interpreted as the amount of workers of type~$x$ available, and the target measure~$\nu(y)$ the amount of work of type~$y$ needed to create the product. The productivity function~$p = -c$ then describes how productive a worker of type~$x$ is at performing task~$y$. The associated optimal transport problem then maximizes the productivity of the company and finds the optimal assignment of jobs. Moving to the non-conservative setting, the mass-change factor~$m(x,y)$ could model the retention rate, i.e.~how many workers of type~$x$ quit if they are forced to do job~$y$. Including the retention rates will result in a more productive company. 
\end{example}

\subsection{The non-conservative Dual Problem.}\label{sec:ncot_dual_problem}

We open this subsection with a derivation that will motivate the formulation of the dual problem to the optimization problem~\eqref{equ:katorovich_problem}. 
Let~$ \mathcal{M}_+ (\mathcal{X} \times \mathcal{Y})$ denote the set of positive finite measures on~$\mathcal{X} \times \mathcal{Y}$. We define the function~$\chi: \mathcal{M}_+ (\mathcal{X} \times \mathcal{Y}) \to [0, \infty]$ as
\begin{align*}
  \chi(\pi)
=&\sup_{\varphi,\psi}\Biggl\{
\int_{\mathcal{X}}\varphi(x)\,\mu(dx)-\int_{\mathcal{X} \times \mathcal{Y}} \varphi(x)\,\pi(dx,dy)
\\
&+\int_{\mathcal{Y}}\psi(y)\,\nu(dy)
\int_{\mathcal{X}\times \mathcal{Y}}m(x,y)\,\pi(dx,dy)
-\int_{\mathcal{X} \times \mathcal{Y}}\psi(y)m(x,y)\,\pi(dx,dy)
\Biggr\},
\end{align*}
where the supremum is over all bounded and continuous functions~$\varphi \in \cC(\mathcal{X})$ and~$\psi \in \cC(\mathcal{Y})$. We observe that the function~$\chi$ coincides with 
\begin{equation}\label{equ:def_chi}
\chi(\pi)=
\begin{cases}\infty, & \text{if } \pi\in\mathcal{M}_+(\mathcal{X}\times 
\mathcal{Y})\setminus\Gamma_{m}(\mu,\nu),\\[1mm]
0, & \text{if } \pi\in\Gamma_{m}(\mu,\nu).
\end{cases}
\end{equation}
Hence, we obtain that 
\begin{align*}
  & \inf_{\pi\in\Gamma_{m}(\mu,\nu)}\int_{\mathcal{X} \times \mathcal{Y}}c(x,y)\,\pi(dx,dy)
 =\inf_{\pi\in\mathcal{M}_+(\mathcal{X} \times \mathcal{Y})}\Biggl\{
\int_{\mathcal{X} \times \mathcal{Y}}c(x,y)\,\pi(dx,dy)+\chi (\pi) \Biggr\} \\
 & =\inf_{\pi \in \mathcal{M}_+(\mathcal{X} \times \mathcal{Y})  }\sup_{\varphi ,\psi}\Biggl\{
 \int_{\mathcal{X} \times \mathcal{Y}}c(x,y)\,\pi(dx,dy) +\int_{\mathcal{X}}\varphi(x)\,\mu(dx) -\int_{\mathcal{X} \times \mathcal{Y}}\varphi(x)\,\pi(dx,dy)\\
 & \qquad  +\int_{\mathcal{Y}}\psi(y)\,\nu(dy)
 \int_{\mathcal{X}\times \mathcal{Y}}m(x,y)\,\pi(dx,dy)
 -\int_{\mathcal{X} \times \mathcal{Y}}\psi(y)m(x,y)\,\pi(dx,dy)
 \Biggr\}. 
\end{align*}
Consequently, 
\begin{align} \label{equ:KP_larger_dual_last_step}
&\inf_{\pi\in\Gamma_{m}(\mu,\nu)}\int c(x,y)\,\pi(dx,dy)
\geq \sup_{\varphi, \psi }\Biggl\{
\int \varphi(x)\,\mu(dx) 
\\ 
&+\inf_{\pi\in\mathcal{M}_+(\mathcal{X} \times \mathcal{Y})}\int \Biggr( c(x,y)  -\varphi(x)-
\left(\psi(y) - \int \psi(b)\,\nu(db) \right) \,m(x,y)
\Biggr)\pi(dx,dy)
\Biggr\}. \nonumber
\end{align}

For a fixed pair $(\varphi, \psi)$, the inner infimum in~\eqref{equ:KP_larger_dual_last_step} equals $-\infty$ unless
\[
c(x,y)\ge \varphi(x)+\Bigl(\psi(y)-\int \psi (b)\,\nu(db)\Bigr)m(x,y) \qquad \forall (x,y)\in \mathcal{X} \times \mathcal{Y}.
\]
This calculation motivates the following definition.

\begin{definition}[Dual Kantorovich problem in non-conservative optimal transport] \label{def:dual_problem}
    Given a cost function~$c: \mathcal{X} \times \mathcal{Y} \to \mathbb{R}$ and a mass-change factor~$m: \mathcal{X} \times \mathcal{Y} \to \mathbb{R}$, we define the admissibility set~$ \mathcal{A} := \mathcal{A}(c,m)$ as all pairs~$(\varphi, \psi) \in \cC (\mathcal{X}) \times \cC(\mathcal{Y})$ such that
    \begin{align}\label{equ:DP_admissability_set}
      c(x,y) \geq \varphi(x)+\Bigl(\psi(y)-\int \psi (b)\,\nu(db)\Bigr)m(x,y) \qquad \forall (x,y) \in \mathcal{X} \times \mathcal{Y}.
    \end{align}
 The dual problem is given by
\begin{align}
  \sup_{(\varphi, \psi) \in \mathcal{A}(c,m)}&
 \int \varphi(x)\,\mu(dx) . \tag{DP} \label{equ:dual_KP} 
\end{align}

\end{definition}
\begin{remark}
    For a constant mass-change factor~$m\equiv1$, one recovers the dual problem of the classical Kantorovich problem 
\begin{align}
         \sup\limits_{\varphi, \psi}&
\left\{ \int \varphi(x)\,\mu(dx) + \int \psi(y)\, \nu(dy) \ \mid  \ \forall x,y: \ c(x,y) \geq \varphi(x) + \psi(y)  \right\}.
\end{align} 
    One difference in non-conservative optimal transport is that the roles of~$\mu$ and~$\nu$ are not symmetric. Consequently, also the roles of~$\varphi$ and~$\psi$ become non-symmetric.
\end{remark}

In light of Definition~\ref{def:dual_problem}, the inequality~\eqref{equ:KP_larger_dual_last_step} provides the weak duality relation summarized below.

\begin{lemma}
For all~$\pi\in\Gamma_{m}(\mu,\nu)$ and all~$(\varphi , \psi ) \in \mathcal{A}$ it holds
\begin{align*}
      \int c(x,y)\,\pi(dx,dy) \geq \inf \eqref{equ:katorovich_problem} \geq  \sup \eqref{equ:dual_KP}  \geq  \int \varphi(x)\,\mu(dx).   
\end{align*}
\end{lemma}

Let us now investigate existence of a maximizer of the dual problem~\eqref{equ:dual_KP}. 

\begin{assumption}\label{asu:existence_optimal_potentials}
  The sets~$\mathcal{X}$ and~$\mathcal{Y}$ are compact and the functions~$c$ and~$m$ are continuous with~$\inf\limits_{x,y}m(x,y)>0$.   
\end{assumption}

\begin{theorem}\label{pro:solution_dual_KP}
 Let Assumption~\ref{asu:existence_optimal_potentials} be satisfied. Then, there exists $(\varphi, \psi)\in \mathcal{A}$ satisfying 
 \begin{align}\label{def:generalized_c_transform}
     \varphi(x) : =   \psi_m^c (x): = \inf_{y \in \mathcal{Y}} \left(c(x,y) - \left( \psi(y) - \int_{\mathcal{Y}} \psi(b) \nu(db) \right) m(x,y) \right)
 \end{align}
 and
 \begin{align}\label{equ:realtion_phi_psi}
        \psi(y) = \inf_{x\in \mathcal{X}} \frac{c(x,y)- \varphi(x)}{m(x,y)} \quad \mbox{with}  \quad \int \psi(y) \nu(dy) = 0.
 \end{align}
 such that
 \begin{align*}
     \sup \eqref{equ:dual_KP} = \max \eqref{equ:dual_KP} = \int \varphi(x)  \mu(dx).
 \end{align*}
\end{theorem}

The main idea of the proof is using the Arzela-Ascoli theorem to show that a maximizing sequence~$(\varphi_n, \psi_n) \in \mathcal{A}$ converges to a maximizer~$(\varphi , \psi) \in \mathcal{A}$. The equation~\eqref{def:generalized_c_transform} takes on a similar role as the $c$-transform in classical optimal transport. In contrast to the classical case, the argument get more involved due to~\eqref{def:generalized_c_transform} and~\eqref{equ:realtion_phi_psi} being a more complicated analogue of of the $c$-transform that lacks symmetry. The compactness of~$\mathcal{X}$ and~$\mathcal{Y}$ is assumed to simplify the presentation and arguments.

\begin{proof}[Proof of Theorem~\ref{pro:solution_dual_KP}]
    In the first part of this proof we show that the feasible set of the dual problem~\eqref{equ:dual_KP} can be restricted to admissible pairs $(\varphi, \psi)\in \mathcal{A}$ satisfying~\eqref{def:generalized_c_transform} and~\eqref{equ:realtion_phi_psi} without changing the optimal value or solution of the problem. 
    We start with deducing~\eqref{def:generalized_c_transform}. For this, let us fix a function~$\psi$ and consider all~$\varphi$ such that~$(\varphi, \psi) \in \mathcal{A}$ is admissible. From the definition of~$\mathcal{A}$ and the continuity of~$c$ and~$m$ it follows that
    \begin{align} \label{equ:reduction_c_transform} 
        \varphi(x) \leq \inf_{y\in \mathcal{Y}} \left(  c(x,y) - \Bigl(\psi(y)-\int \psi (b)\,\nu(db)\Bigr)m(x,y) \right) =: \psi_{m}^c(x).
    \end{align}
    As all involved functions are continuous, the function~$\psi_m^c$ is also continuous. Because~$\mathcal{X}$ is bounded, the function~$\psi_m^c$ is also bounded. It follows that the pair~$(\psi_{m}^c, \psi) \in \mathcal{A}$ and from~\eqref{equ:reduction_c_transform} it follows that
    \begin{align*}
        \int_{\mathcal{X}} \varphi(x) \mu(dx) \leq \int_{\mathcal{X}}  \psi_{m}^c (x) \mu(dx).
    \end{align*}
    Therefore, we can w.l.o.g.~impose~$\varphi = \psi_{m}^c$. \\
    
    Next, let us show that we can w.l.o.g.~impose~\eqref{equ:realtion_phi_psi}. Indeed, from the admissibility condition it follows that 
\begin{align*}
    \psi(y) - \int \psi(y) \nu(dy)  \leq  \inf_x \frac{c(x,y)- \varphi(x)}{m(x,y)}.
\end{align*}
Integrating this inequality yields
\begin{align*}
    0 = \int_{\mathcal{Y}} \psi(y) \nu (dy) - \int_{\mathcal{Y}} \psi(b) \nu (db)  \leq \int_{\mathcal{Y}} \inf_x \frac{c(x,y)- \varphi(x)}{m(x,y)} \nu(dy).
\end{align*}
By continuity and monotonicity there exists a constant~$k\geq0$ such that
\begin{align*}
    0 = \int_{\mathcal{Y}} \inf_x \frac{c(x,y)- \varphi(x)-k}{m(x,y)} \nu(dy).
\end{align*}
We define now 
\begin{align*}
    \bar \varphi = \varphi + k \qquad \mbox{and} \qquad \bar \psi (y) = \inf_x \frac{c(x,y)- \varphi(x)-k}{m(x,y)}.
\end{align*}
It follows from the construction that~$(\bar \varphi, \bar \psi) \in \mathcal{A}$ is admissible and
\begin{align*}
    \int \varphi \  d\mu \leq  \int \bar \varphi \ d \mu,
\end{align*}
which verifies that we can w.l.o.g. impose~\eqref{equ:realtion_phi_psi}. \\

In the next part of this proof, we show that admissible pairs~$(\varphi, \psi) \in \mathcal{A}$ satisfying~(\ref{def:generalized_c_transform}--\ref{equ:realtion_phi_psi}) are uniformly bounded, i.e.~there is a constant~$K>0$ such that
\begin{align}\label{equ:uniform_boundedness_phi_psi}
    -K \leq \varphi(x) \leq K, \quad  -K \leq \psi(y) \leq K \qquad \mbox{for all $x\in \mathcal{X}, \ y \in \mathcal{Y}$}.
\end{align}
\begin{enumerate}
    \item[(i)]  From~\eqref{equ:realtion_phi_psi} it follows that $\psi(y^*) \geq 0$ for some~$y^* \in \mathcal{Y}$. Using~\eqref{equ:reduction_c_transform} and~$\int \psi \ d \nu =0$, it follows that~$\varphi$ is uniformly bounded from above, namely
\begin{align}\label{equ:varphi_upper_bound}
    \varphi(x) \leq c(x,y^*) \leq \sup_{x,y} c(x,y) < \infty.  
\end{align}

\item[(ii)] From~\eqref{equ:varphi_upper_bound} it follows that
\begin{align*}
    \frac{c(x,y) - \varphi(x)}{m(x,y)} \geq \frac{\inf_{x,y} c(x,y) - \sup_{x,y} c(x,y)}{\inf_{x,y} m(x,y)},
\end{align*}
which together with~\eqref{equ:realtion_phi_psi} yields  that~$\psi$ is uniformly bounded from below.

\item[(iii)] Note that the objective value for an admissible pair $(\varphi, \psi) \in \mathcal{A}$ with $\varphi(x) < \inf\limits_{x,y} c(x,y)$ across all $x \in \mathcal{X}$ is dominated by an objective value for the constant admissible pair $(\bar{\varphi}, \bar{\psi}) \equiv (\inf\limits_{x,y} c(x,y), 0)$. Therefore, we can w.l.o.g.~exclude such admissible pairs and assume that there exists some~$x^* \in \mathcal{X}$ such that~$\varphi(x^*) \geq \inf\limits_{x,y} c(x,y)$. In light of this, for all $y \in \mathcal{Y}$ we obtain form~\eqref{equ:realtion_phi_psi}
\begin{align*}
    \psi(y) \leq \frac{c(x^*, y) - \varphi(x^*)}{m(x^*, y)} \leq \frac{\sup_{x,y} c(x,y) - \inf_{x,y} c(x,y)}{\inf_{x,y} m(x,y)}.
\end{align*}
This shows that~$\psi$ is uniformly bounded from above.

\item[(iv)] Finally, the uniform upper bound on~$\psi$ implies by~\eqref{def:generalized_c_transform} and~\eqref{equ:realtion_phi_psi} a uniform lower bound on~$\varphi$.
\end{enumerate}

In the last part of the proof, we argue existence of an optimal solution. By the previous arguments, we restrict ourselves to admissible pairs satisfying~(\ref{def:generalized_c_transform}--\ref{equ:realtion_phi_psi}) and uniform boundedness. Let us consider a maximizing sequence~$(\varphi_n , \psi_n) \in \mathcal{A}$ realizing the supremum in~\eqref{equ:dual_KP}. In order to apply Arzela-Ascoli, it is left to show that the functions $\varphi_n$ and~$\psi_n$ are equi-continuous. Let us first consider~$\{\varphi_n\}$ and define the family of functions~$f_{y,n} (x) : = c(x,y) - \psi_n(y) m(x,y)$ and observe $\varphi_n(x) = \inf_{y} f_{y,n} (x)$. It is a known fact that modulus of continuity of~$f_{y,n}(x)$ carries over to the~$\inf_{y} f_{y,n} = \varphi_n$ (see e.g.~\cite[Box 1.5]{santambrogio2015otam}). Direct calculation using~\eqref{equ:realtion_phi_psi} and~\eqref{equ:uniform_boundedness_phi_psi} gives
\begin{align}\label{equ:modulus_continuity}
    |f_{y,n} (x_2) - f_{y,n} (x_1)| \leq | c(x_2,y) - c(x_1,y)| + K  |m(x_2,y) - m(x_1,y)|  .
\end{align}
So the modulus of continuity of the functions~$\varphi_n$ only depends on the modulus of continuity of~$c$ and~$m$. Therefore the functions~$\varphi_n$ are equi-continuous. 
Let us now consider the family~$\left\{ \psi_n \right\}$ and define the auxiliary functions $g_{n,x} = \frac{c(x,y) - \varphi_n(x)}{m(x,y)}$. As the family~$\varphi_n$ is equi-continuous with modulus of continuity only depending on~$c$ and~$m$, a direct calculation shows that the family~$g_{n,x}$ also has a uniform modulus of continuity, and therefore $\inf_{x} g_n = \psi_n$ has a uniform modulus of continuity as well. Hence, the functions~$\psi_n$ are equi-continuous.

As a consequence, we can apply Arzela-Ascoli and can extract a subsequence~$(\varphi_{n_k}, \psi_{n_k})$ that converges uniformly to a maximizer~$(\varphi, \psi) \in \mathcal{A}$.
\end{proof}

\subsection{Strong Duality} \label{sec:strong_duality}

In this section we show that under the Assumption~\ref{asu:existence_optimal_potentials} strong duality holds (see Theorem~\ref{the:strong_duality} below). We generalize the argument outlined in~\cite[Section~1.6.2]{santambrogio2015otam}, which is based on an argument by C.~Jimenez and is adapted from~\cite[Section~4]{bouchitte2001characterization}.

\begin{definition}
	For all $p\in \cC(\mathcal X \times \mathcal Y)$, we define
	\[
	H(p) = - \sup_{\varphi, \psi} \left\{ \int \varphi \ d\mu : c(x,y) - p(x,y) \geq \varphi(x) + \left( \psi(y) - \int \psi(b) d\nu(b) \right) m(x,y) \right\},
	\]
    where the supremum is taken over $\varphi\in \cC(\mathcal X), \ \psi\in \cC(\mathcal Y)$.
\end{definition}

Notice that~$\sup \eqref{equ:dual_KP} = -H(0)$. Moreover, note than under the Assumption~\ref{asu:existence_optimal_potentials} the supremum is attained, see Theorem~\ref{pro:solution_dual_KP}.

\begin{lemma}\label{lemma:Hconvexlsc}
	Under the Assumption~\ref{asu:existence_optimal_potentials}, the function $H: \cC(\mathcal X \times \mathcal Y) \to \mathbb R \ \cup \ \{\pm \infty\}$ is proper, convex and l.s.c.~with respect to uniform convergence on $\mathcal X \times \mathcal Y$. 
\end{lemma}
\begin{proof}

For any continuous function $p\in \cC(\mathcal X \times \mathcal Y)$, the admissible set $\mathcal{A} (c - p, m)$ and its related dual problem satisfy the Assumption~\ref{asu:existence_optimal_potentials}. Since by Theorem~\ref{pro:solution_dual_KP} the dual problem with cost function $c-p$ admits optimal potentials, $H(p)$ is real-valued and $H$ is a proper.

    Let $p_1, p_2 \in  C(\mathcal X \times \mathcal Y)$ and $t\in [0,1]$. Observe that, for any $(\varphi_1, \psi_1)$ feasible for $p_1$ and any $(\varphi_2, \psi_2)$ feasible for $p_2$, the convex combination $(t\varphi_1 + (1-t) \varphi_2, t \psi_1 + (1-t)\psi_2)$ is admissible for $tp_1 + (1-t)p_2$. Given the linearity of the objective $\int \varphi d\mu$, it follows that $H(tp_1 + (1-t)p_2) \geq tH(p_1) + (1-t) H(p_2)$ and $H$ is convex.

 For semi-continuity, suppose that we have $p_n \to p$ uniformly on $\mathcal X \times \mathcal Y$. Let us take a subsequence $n_k$ such that $H(p_{n_k}) \to \liminf H(p_{n})$. By the converse of the Arzela-Ascoli theorem (see~\cite[Theorem 7.25]{rudin1976principles}), $\{p_n\}$ and therefore $\{p_{n_k}\}$ are equicontinous and bounded. Let $\varphi_{n_k}, \psi_{n_k}$ be the corresponding optimal potentials for $p_{n_k}$. Recall from the proof of Theorem~\ref{pro:solution_dual_KP} that the modulus of continuity of the potential functions depends only on the modulus of the continuity of $c - p_{n_k}$ and $m$. Since $\{p_{n_k}\}$ is equicontinous, then we know also that we may uniformly bound the modulus of continuity of the potentials $\varphi_{n_k}, \psi_{n_k}$ so these are also bounded and equicontinuous. So taking further subsequence via Arzela-Ascoli, without altering notation we have $\varphi_{n_k} \to \varphi, \psi_{n_k}\to \psi$ uniformly. Notice that these limits $(\varphi, \psi)$ will clearly be admissible for $p$, and
\[
H(p) \leq - \int \varphi \ d\mu  = \lim_{k\to \infty} H(p_{n_k}) = \liminf_{n\to \infty} H(p_n)
\]
completes the argument.
\end{proof}

If we endow the space $\cC(\mathcal X \times \mathcal Y)$ with the topology induced by the sup-norm, we know by the Riesz–Markov representation theorem that its dual space is the space of finite signed Radon measures $\mathcal M (\mathcal X \times \mathcal Y)$. With this in mind, the Legendre conjugate function $H^* : \mathcal M (\mathcal X \times \mathcal Y) \to \mathbb R \cup \{+\infty\}$ is given by
\begin{align*}
    H^*(\pi) = &\sup_{p \in C(\mathcal X \times \mathcal Y)} \sup_{\varphi, \psi} \; \Big\{ \int_{\mathcal X \times \mathcal Y} p\ d\pi  + \int_{\mathcal{X}} \varphi \ d\mu \ : \\& \quad   c(x,y) - p(x,y) \geq \varphi(x) + \Big( \psi(y) - \int \psi(b) d\nu(b) \Big) m(x,y) \Big\}.	
\end{align*}
Notice that for $\pi \not \in \mathcal M_+(\mathcal X \times \mathcal Y)$ there exists some $p_0\leq 0$ with $\int p_0 d\pi > 0$, and the constant pair $(\varphi, \psi) = (\inf_{x,y} c(x,y), 0)$ is admissible for the second supremum. Therefore, since $p_0$ can be scaled by arbitrary positive scalar, it holds $H^*(\pi) = \infty$ for $\pi \not \in \mathcal M_+(\mathcal X \times \mathcal Y)$. 

Let us then consider $\pi \in \mathcal M_+(\mathcal X \times \mathcal Y)$. For a fixed pair $(\varphi, \psi)$, the supremum over $p \in \cC(\mathcal X \times \mathcal Y)$ is attained by
\[
p(x,y) = c(x,y) - \varphi(x) - \left(\psi(y) - \int \psi(b) d\nu(b)\right)m(x,y).
\]
Consequently, for $\pi \in \mathcal M_+(\mathcal X \times \mathcal Y)$ the Legendre conjugate becomes 
\begin{align*}
  H^*(\pi) &= \int c(x,y) d\pi(x,y)+ \sup_{\varphi,\psi} \int\varphi d\mu -\int \varphi(x)d\pi(x,y)\\ & \quad \;  - \int \left(\psi(y) - \int \psi(b) d\nu(b)\right)m(x,y) d\pi(x,y) \\
  &= \int c(x,y) d\pi(x,y) + \chi(\pi).
\end{align*}
For the proof of the next theorem we recall the Fenchel–Moreau theorem (cf.~\cite[Theorem 12.2]{rockafellar1970convex} or~\cite[Theorem 3.2.2]{CorreaRafael2023FoCA}). With that in hand we will prove strong duality.

\begin{theorem}[Fenchel–Moreau] 
Let $X$ be an  Hausdorff, locally convex space, then for any proper, lower semi-continuous, and convex function $f : X \to \mathbb R \ \cup \ \{\pm \infty\}$, we have that $f^{**}=f$.
\end{theorem} 

\begin{theorem}[Strong Duality]\label{the:strong_duality}
If Assumption~\ref{asu:existence_optimal_potentials} holds and $\mathcal X, \mathcal Y$ are convex, then $\inf \eqref{equ:katorovich_problem} = \sup \eqref{equ:dual_KP}$.
\end{theorem}
\begin{proof}
    In view of Lemma~\ref{lemma:Hconvexlsc}, the function $H$ satisfies the assumptions of the Fenchel-Moreau theorem. Therefore, it follows
    \begin{align*}
	\sup \eqref{equ:dual_KP} &= -H(0) = -H^{**}(0) = - \sup_{\pi \in \mathcal M (\mathcal X \times \mathcal Y)}  \left\lbrace \langle 0, \pi \rangle - H^*(\pi) \right\rbrace \\
    & = \inf_{\pi \in \mathcal M (\mathcal X \times \mathcal Y)} H^*(\pi) = \inf_{\pi \in \mathcal{M}_+ (\mathcal X \times \mathcal Y)} \int c(x,y) d \pi(x,y) +\chi (\pi) \\
    & = \inf_{\pi \in \Gamma_m (\mu, \nu)} \int c(x,y) d \pi(x,y) = \inf \eqref{equ:katorovich_problem}.
	\end{align*}
    Here we used $H^*(\pi) = \infty$ for $\pi \not \in \mathcal M_+(\mathcal X \times \mathcal Y)$ shown above, and $\xi(\pi) = \infty$ for $\pi \not \in \Gamma_m (\mu, \nu)$ given in~\eqref{equ:def_chi}.
\end{proof}

\subsection{Existence of optimal transport map(s)}\label{sec:existence_optimal_transport_map}

In this subsection, we study whether the optimal transport \textit{plan} of problem~\eqref{equ:katorovich_problem} can be expressed via optimal transport \textit{map(s)}. As the roles of the measure~$\mu$ and~$\nu$ are not symmetric, we define two notions of optimal transport maps. Throughout, $\sharp$ denotes the push-forward measure.

\begin{definition}[Primal and dual optimal transport map] \label{def:optimal_transport_maps}
    Assume that~$\pi^*$ is an optimal transport plan for the problem~\eqref{equ:katorovich_problem}. 
    \begin{enumerate}
        \item[(i)] A function ~$T: \mathcal{X} \to \mathcal{Y}$ is called (primal) optimal transport map if $\pi^*= (\Id_X , T)_{\sharp} \mu$  is given by the push forward of~$\mu$ under the map~$(\Id_{\mathcal{X}}, T)$.
        \item[(ii)] A function~$S: \mathcal{Y} \to \mathcal{X}$ is called dual optimal transport map of~\eqref{equ:katorovich_problem} if $\pi^*= (S , \Id_Y)_{\sharp} \lambda$. Here, the probability measure $\lambda$ is given by the Radon-Nikodym derivative 
        \begin{align*}
            \lambda(dy) =\frac{1}{Z} \frac{1}{m(S(y),y)} \nu(dy),
        \end{align*}
        where~$Z$ denotes the normalization constant transforming $\lambda$ into a probability measure. 
    \end{enumerate}
\end{definition}

\begin{proposition}[Criterion for the existence of optimal transport maps]\label{the:characterization_optimal_maps}
Let~$\pi^*$ be the optimal transport plan for problem~\eqref{equ:katorovich_problem}.
\begin{enumerate}
    \item[(i)] If for a function~$T: \mathcal{X} \to \mathcal{Y}$ it holds 
    \begin{align} \label{eq:support1}
        \supp \pi^* \subset \left\{ (x,T(x)) \ | \ x \in \mathcal{X} \right\},
    \end{align}
    then~$T$ is an optimal transport map of~\eqref{equ:katorovich_problem}.

    \item[(ii)] If for a function~$S: \mathcal{Y} \to \mathcal{X}$ it holds 
    \begin{align} \label{eq:support2}
        \supp \pi^* \subset \left\{ (S(y),y) \ | \ y \in \mathcal{Y} \right\},
    \end{align}
    then~$S$ is a dual optimal transport map of~\eqref{equ:katorovich_problem}.
\end{enumerate}
\end{proposition}    
    
\begin{proof}[Proof of Proposition~\ref{the:characterization_optimal_maps}]

\begin{enumerate}
    \item[(i)] Let~\eqref{eq:support1} hold. We start with disintegrating the probability measure~$\pi^*(x,y)$ into its conditional measures~$\pi(dy|x)$ and marginal~$\bar \pi(dx)$, i.e.
        \begin{align*}
            \pi^* (dx, dy) = \pi(dy|x) \bar \pi(d x).
        \end{align*}
        As~$\pi^* \in \Gamma_m(\mu, \nu)$ the marginal measure is given by~$\bar \pi(d x) = \mu(d x)$. If the support of~$\pi^*$ is contained in the set $\left\{ (x,T(x)) \ | \ x \in \mathcal{X} \right\}$, then the support of the conditional probability measure~$\pi(dy|x)$ is concentrated on a single point~$\left\{ T(x)\right\}$. Hence,~$\pi(dy|x) = \delta_{T(x)}$. This yields the desired representation~$\pi^* = (\Id_x , T)_{\sharp} \mu$.

        \item[(ii)] Let~\eqref{eq:support2} hold, we follow a similar argument as in the first case. This time we disintegrate~$\pi^*$ into
        \begin{align*}
            \pi^* (dx, dy) = \pi(dx| y) \bar \pi(d y).
        \end{align*}
        Again, we get that the conditional measure is given by~$\pi(dx|y) = \delta_{S(y)}$. It is left to verify that the marginal measure is given by $\bar \pi (dy) = \lambda(dy)$, which follows from the~$\pi^* \in \Gamma_{m} (\nu, \mu)$ after a short calculation.
\end{enumerate}
\end{proof}

    \begin{remark}\label{rem:uniqueness_via_optimal_transport_maps}
        If the cost function is convex, then (primal and dual) optimal transport map must be~$\mu$-a.e. unique (or~$\nu$-a.e.~respectively). Indeed, if~$\pi_1 := (\Id_{\mathcal{X}}, T_1)_{\sharp} \mu$ and $\pi_2 := (\Id_{\mathcal{X}}, T_2)_{\sharp} \mu$ are optimal, then by convexity $\pi_3 = \frac{1}{2} \pi_1 + \frac{1}{2} \pi_2$ is also optimal. The optimal transport plan~$\pi_3$ can also be represented via an optimal transport map~$T_3$, i.e.~$\pi_3 = (\Id_{\mathcal{X} }, T_3)_{\sharp} \mu$, which is only possible if~$T_1= T_2$, $\mu$-a.e.. If every optimal transport plan is given by an optimal transport map, this implies a.s.~uniqueness of the optimal transport plan as well (cf.~Proposition~\ref{the:existence_optimal_maps} and Proposition~\ref{the:solution_quadratic_leaky_monge} below).  
    \end{remark}
    \begin{remark}
        If there exists an optimal transport map~$T$ as well as a dual optimal transport map~$S$, then it holds~$(x,T(x))= (S(y), y)$ for~$\pi^*$-a.e.~$(x,y)$. Therefore, $T$ and $S$ are inverse to each other.
    \end{remark}

    From Proposition~\ref{the:characterization_optimal_maps} it becomes clear that in order to show the existence of optimal transport maps, one needs to study the support of the optimal transport plan.

 \begin{proposition}\label{the:characterization_support_optimal_transport_plan}
          Assume that~$\min \eqref{equ:katorovich_problem} = \max \eqref{equ:dual_KP}$, which means that there is~$\pi^* \in \Gamma_{m} (\mu, \nu)$ and~$(\varphi, \psi) \in \mathcal{A}$ such that
    \begin{align}\label{equ:strong_duality_specific}
       \int \varphi(x) \  \mu(dx) = \int c(x,y)  \ d  \pi^*.
    \end{align}
        Then, the support of~$\pi^*$  is contained in the set
        \begin{align}\label{equ:support_pi_contained_in_c_convex_set}
            \cS :=  \left\{ (x,y) \in \mathcal{X} \times \mathcal{Y} \ | \ c(x,y) = \varphi(x) + \left( \psi(y) - \int \psi \ d \nu \right) m(x,y) \right\},
        \end{align} 
        i.e. $\supp \pi^* \subset \cS$. 
        Moreover, consider the family of functions~$f_x: \mathcal{Y} \to \mathbb{R}$ and~$g_y: \mathcal{X} \to \mathbb{R}$ given by 
        \begin{align}\label{equ:criterion_optimal_transport_map}
            f_x(y):= c(x,y) - \left(\psi(y)- \int \psi \ d \nu \right) m(x,y)
        \end{align}
        and
        \begin{align}\label{equ:criterion_dual_optimal_transport_map}
            g_y(x):= \frac{ c(x,y) - \varphi(x)}{m(x,y)}.
        \end{align}
        Then, the set~$\cS$ is given by
        \begin{align}\label{equ:support_pi_characterized via minima}
            \cS = \left\lbrace (x,y) \in \mathcal{X} \times \mathcal{Y} \ | \ y \mbox{ minimizes } f_x \mbox{ and } x \mbox{ minimizes } g_y \right\rbrace.
        \end{align}
        
    \end{proposition}

    \begin{proof}
    Recall from the proof of Theorem~\ref{pro:solution_dual_KP} that an optimal dual solution~$(\varphi, \psi) \in \mathcal{A}$ satisfies conditions~(\ref{def:generalized_c_transform}--\ref{equ:realtion_phi_psi}). Then, in view of definitions of $f_x$ and $g_y$ in~(\ref{equ:criterion_optimal_transport_map}--\ref{equ:criterion_dual_optimal_transport_map}), it holds
    \begin{align}\label{equ:assumption_kantorovich_potentials}
        \varphi(x) = \inf_y f_x(y) \quad \mbox{ and } \quad \psi(y) = \inf_x g_y(x).
    \end{align}
    Integrating the admissibility condition \eqref{equ:DP_admissability_set} with respect to~$\pi$ implies
    \begin{align*}
        \int & c(x,y) \pi(dx, dy)  \geq \int \left( \varphi(x) + \left( \psi(y) - \int \psi(y) \ \nu(dy) \right) m(x,y) \right) \ \pi(dx, dy) \\
        & = \int \varphi(x) \ \mu(dx) + \int \psi(y) m(x,y) \pi(dx, dy) - \int \psi(y) \ \nu(dy) \int  m(x,y)  \ \pi(dx, dy) \\
        & = \int \varphi(x) \ \mu(dx),
    \end{align*}
    where we used in the last step the second marginal condition on~$\pi$ (cf.~\eqref{equ:def_gamma_m}). Hence, the hypothesis~\eqref{equ:strong_duality_specific} jointly with admissibility \eqref{equ:DP_admissability_set} implies the desired statement~\eqref{equ:support_pi_contained_in_c_convex_set}. The second characterization of the set~$\cS$ follows directly from~\eqref{equ:assumption_kantorovich_potentials}.
    \end{proof}

The theory for the existence of optimal transport maps does not easily carry over from the conservative case and needs further investigation. The next statement show the existence of (primal or dual) optimal transport maps if the mass-change factor~$m \approx 1$, i.e.~if the model is close to a classical optimal transport problem. 

\begin{proposition}[Existence of optimal transport maps in the perturbative regime]\label{the:existence_optimal_maps}

    Let $\mathcal{X}, \mathcal{Y} \subset \mathbb{R}^n$ be non-empty, compact, connected sets that are equal to the closure of their interiors (i.e.~$\mathcal{X}, \mathcal{Y}$ are domains). Let~$h,d: [0, \infty) \to \mathbb{R}$ be two smooth, convex, Lipschitz functions such that 
    \begin{itemize}
        \item $h$ is strictly increasing and $\lim_{t \downarrow 0} \frac{h(t) - h(0)}{t} =0$;
        \item $d$ is non-decreasing and $\lim_{t \downarrow 0} \frac{d(t) - d(0)}{t} =0$.
    \end{itemize}
    Additionally, we assume that there is~$\varepsilon>0$
    \begin{align}\label{equ:assumptions_d_and_h}
    \hspace{-0.5cm}
             h''(z)  > \varepsilon d''(z) \quad \mbox{and} \quad
    h'(z) > \varepsilon d'(z)  \qquad \mbox{for all } 0< z \leq \sup_{x\in \mathcal{X} , y \in \mathcal{Y}} |x-y|.      
    \end{align}
    We consider the cost function~$c(x,y):= h(|x-y|)$ and the mass-change factor~$m(x,y)= \exp(- k \ d(|x-y|))$ where $k \geq 0$ is a constant. \newline

    Then, there is a constant~$K$ such that for all $0 \leq k < K$ it holds:
    \begin{enumerate}
    \item[(i)] If the probability measure $\mu \in \mathcal{P} (\mathcal{X})$ absolutely continuous w.r.t.~the Lebesgue measure and the boundary of~$\mathcal{X}$ has Lebesgue-measure 0, then every optimal transport plan is given by a unique optimal transport map~$T: \mathcal{X} \to \mathcal{Y}$.
    
    \item[(ii)] If the probability measure $\nu \in \mathcal{P} (\mathcal{Y})$ absolutely continuous w.r.t.~the Lebesgue measure, and the boundary of~$\mathcal{Y}$ has Lebesgue-measure 0, then then every optimal transport plan is given by a unique dual optimal transport map~$S: \mathcal{Y} \to \mathcal{X}$.
    \end{enumerate}

\end{proposition}

\begin{proof}[Proof of Proposition~\ref{the:existence_optimal_maps}]
Under the current assumptions, Theorems~\ref{the:solution_KP},~\ref{pro:solution_dual_KP}, and~\ref{the:strong_duality} hold. Consequently, the optimal dual solution $(\varphi, \psi) \in \mathcal{A}$ exists and the assumptions of Proposition~\ref{the:characterization_support_optimal_transport_plan} are satisfied. 
Hence, for any point~$(x,y)\in \supp \pi^*$ it holds that the value~$x$ is a minimizer of the function $g_y(x) = \frac{c(x,y)- \varphi(x)}{m(x,y)}$ and the value~$y$ is a minimizer of the function $f_x(y) = c(x,y) -  \psi(y) m(x,y)$. We now show that~$f_x$ and~$g_y$ are a.e.~differentiable.\smallskip

Recall from the proof of Theorem~\ref{pro:solution_dual_KP} that $\varphi, \psi$ are uniformly bounded. Specifically, we can show that 
\begin{align}\label{equ:uniform_bounds_psi_phi}
-\bar{h}(1 + 2 e^{2k\bar{d}}) \leq \varphi (x) \leq \overline{h}  \quad \text{ and } \quad -2 \bar{h} e^{k \bar{d}}  \leq \psi(y) \leq 2 \bar{h} e^{k \bar{d}}    
\end{align} 
where $\overline{h} := \sup\limits_{x \in \mathcal{X}, y \in \mathcal{Y}} | h(|x - y|) |$ and $\overline{d} := \sup\limits_{x \in \mathcal{X}, y \in \mathcal{Y}} |d(|x - y|)|$. In particular those bounds are uniform in~$|k| \leq K$. As~$\mathcal{X}$ and~$\mathcal{Y}$ are compact and additionally~$\varphi$ and~$\psi$ are uniformly bounded, the functions~$\varphi$ and~$\psi$ inherit the modulus of continuity, up to a universal constant, from~$c$ and~$m$ (see proof of Theorem~\ref{pro:solution_dual_KP} for more details). Therefore~$\varphi$ and~$\psi$ are both Lipschitz, which implies together with our assumptions on~$c$ and~$m$ that~$f_x$ and~$g_y$ are also Lipschitz (see~\eqref{equ:criterion_optimal_transport_map} and~\eqref{equ:criterion_dual_optimal_transport_map}). It follows from Rademacher's theorem (here we use that~$\mathcal{X}$ is the closure of its interior) that if~$\mu$ is absolutely continuous w.r.t.~the Lebesgue measure, the functions~$g_y$ are~$\mu$-a.e.~differentiable. Similar the functions~$f_x$ are~$\nu$-a.e.~differentiable if~$\nu$ is absolutely continuous w.r.t.~the Lebesgue measure. \smallskip

Let us turn to the verification of~(i).
Recall from the characterization~\eqref{equ:support_pi_characterized via minima} that for all $(x,y) \in \supp \pi^*$, the value $x$ minimizes the function $g_y$. For~$(x,y) \in \supp \pi^*$ it holds (recall that boundary points are negligible by assumption)
\begin{align}\label{equ:characteristic_equation_optimal_map}
    0 &= m(x,y) \nabla_x g_y(x)  \\ & =  \nabla_x c(x,y) - \nabla_x \varphi (x) - \Big(  c(x,y) - \varphi (x)  \Big)  \nabla_x   \ln m(x,y) \qquad \qquad \mbox{$\mu$ a.e.~$x$.} 
\end{align}
For~$c(x,y)= h(|x-y|)$ and~$m(x,y)= \exp (-k d(|x-y|)) >0$ we obtain
\begin{align}\label{equ:characteristic_equation_1}
    \Big(  h'(|x-y|)   + k  d'(|x-y|) \left(h(|x-y|) - \varphi(x) \right) \Big) \frac{x-y}{|x-y|} =  \nabla_x \varphi (x).
\end{align}

We now argue that after keeping~$x$ fix the equation~\eqref{equ:characteristic_equation_1} uniquely characterizes~$y$. 
For domains $\mathcal{X}, \mathcal{Y}$ there exists a finite $D:= \sup_{x,y} |x-y|$ and since~$(x,y)$ solves the equation~\eqref{equ:characteristic_equation_1}, it must hold that~$x-y = z \ \frac{\nabla_x \varphi(x)}{|\nabla_x \varphi(x)|}$ for some value $-D \leq z \leq D$. In light of~\eqref{equ:characteristic_equation_1}, the scalar $z$ must satisfy
\begin{align}\label{eq:T(z)}
    \frac{z}{|z|} \underbrace{\Big(  h'(|z|) +  k  d'(|z|) \left(h(|z|) - \varphi(x) \right)  \Big)}_{=:F(|z|)} = |\nabla_x \varphi(x)|.
\end{align}
In view of the uniform upper bound $\varphi(x) \leq \bar{h}$ and $z \in [-D, D]$ it holds $h(|z|) - \varphi(x)  \geq -2\bar{h}$. Then, for $0 < k < \frac{\varepsilon}{2 \bar{h}}$ it holds for $|z| > 0$
\begin{align*}
    F(|z|) =  h'(|z|) +  k \  d'(|z|) \left(h(|z|) - \varphi(x) \right)  \geq  h'(|z|) - 2 \bar{h} \ k \  d'(|z|) \overset{\eqref{equ:assumptions_d_and_h}}{>} 0.
\end{align*}
Additionally, it follows form our assumptions that~$F(|z|) = 0 $ if and only if~$|z|=0$. Therefore, in the case~$|\nabla_x \varphi(x)| =0$, it follows that~$y=x$ is uniquely characterized by~$x$. \smallskip

In the case~$|\nabla_x \varphi(x)| \neq 0$, the equation~\eqref{eq:T(z)} may only have solutions $z > 0$ for sufficiently small $k >0$ . We  now show that $F(z)$ is strictly increasing for sufficiently small $k >0$. In view of the assumptions, it holds that
\begin{align}
    F'(z) & = h''(z) +  k   d''(z) \left(h(z) - \varphi(x) \right) + k d'(z) h'(z) \\
    & \geq  h''(z) - 2 \bar{h} k d''(z) > (\varepsilon - 2 \bar{h} k) d''(z) \geq 0,
\end{align} 
for any~$0\leq k\leq \frac{\varepsilon}{2 \bar{h}}$.
This concludes the argument for $y$ being uniquely characterized via~\eqref{eq:T(z)} given $x$, which proves part (i).

The argument for~(ii) follows the footstep of~(i). We have seen above that the functions~$f_x$ are~$\nu$-a.e.~differentiable.  Taking the gradient of~$f_x$ wrt.~$y$  yields for~$(x,y) \in \supp \pi^*$ (boundary points are again negligible by assumption)
\begin{align}\label{equ:characteristic_equation_dual_optimal_map}
    \nabla_y c(x,y) - m(x,y) \ \nabla_y \psi (y)   - \psi (y)  \nabla_y  m(x,y) = 0 \qquad \qquad \mbox{$\nu$ a.e.~$y$.}
\end{align}
    Plugging in the definition of~$c$ and~$m$ and rearranging terms yields
    \begin{align}
        \frac{x-y}{|x-y|} \left( h'(|x-y|) e^{k \ d(|x-y|)} + k \psi(y)  d'(|x-y|)  \right) = - \nabla_y \psi(y).
    \end{align}
    Using the Ansatz~$x-y = z \frac{\nabla_y \psi(y)}{|\nabla_y \psi(y)|}$ yields
\[
\frac{z}{|z|}  \left(  h'(|z|  ) e^{k d(|z| )}   + k  \psi (y)    d'(|z|)  \right) = - | \nabla_y \psi(y)|.
\]
    Due to~\eqref{equ:uniform_bounds_psi_phi}, the function~$|\psi|$ is uniformly bounded independent of~$k$. Similar to the argument in part (i) one can show that for~$0<k$ small enough
    \begin{align*}
       \left(  h'(|z|  ) e^{k d(|z| )}   + k  \psi (y)    d'(|z|)  \right) > 0~ \qquad \mbox{ if }  |z| >0,
    \end{align*}
    and the let hand side can only be zero at~$|z|=0$. This means in the case~$| \nabla_y \psi(y)| =0$, there is only one valid choice~$x=y$. In the case~$| \nabla_y \psi(y)| \neq 0$, the solutions require~$z < 0$. Now, one can argue that for fixed~$y$ the function in the bracket is strictly monotone increasing in~$|z|$ provided~$k$ small enough and therefore has at most one solution. This characterizes~$x$ uniquely in terms of~$y$ and closes the argument.  
\end{proof}

In the proof of Proposition~\ref{the:existence_optimal_maps}, the equations~\eqref{equ:characteristic_equation_optimal_map} and~\eqref{equ:characteristic_equation_dual_optimal_map} play a special role. They are used to characterize the existence of optimal transport maps motivating the following definition. 
\begin{definition}
     Let us consider a pair~$(\varphi, \psi)$ of optimal Kantorovich potentials solving~\eqref{equ:dual_KP} . Then we call the equation~\eqref{equ:characteristic_equation_optimal_map} characteristic equation for the optimal transport map. Similarly, the equation~\eqref{equ:characteristic_equation_dual_optimal_map} is called characteristic equation for the dual optimal transport map. 
\end{definition}

Now, let us illustrate with an example that optimal transport maps can also exists when the mass-change factor~$m$ is not close to~$1$, i.e.~when one does not expect the non-conservative optimal transport to be close to a classical optimal transport. 

\begin{proposition}[Existence of optimal transport maps in the quadratic leaky-Monge problem]\label{the:solution_quadratic_leaky_monge}
     Let $\mathcal{X}, \mathcal{Y} \subset \mathbb{R}^n$ be non-empty, compact, connected sets that are equal to the closure of its interior (i.e.~$\mathcal{X}, \mathcal{Y}$ are domains). We consider the quadratic cost function~$c(x,y):= \frac{1}{2} |x-y|^2$ and the mass-change factor~$m(x,y)= 1- \frac{k}{2}|x-y|^2$ for a constant~$k>0$ such that~$\inf_{x \in \mathcal{X} y \in \mathcal{Y}} m (x,y) > 0$. Then it holds:
    \begin{enumerate}[label=(\roman*)]
        \item If the probability measure $\mu$ absolutely continuous w.r.t.~the Lebesgue measure and the boundary of~$\mathcal{X}$ has Lebesgue-measure 0, then there exists an optimal transport map~$T: \mathcal{X} \to \mathcal{Y}$.
        \item If the probability measure $\nu$ absolutely continuous w.r.t.~the Lebesgue measure and the boundary of~$\mathcal{Y}$ has Lebesgue-measure 0, then there exists a dual optimal transport map~$S: \mathcal{Y} \to \mathcal{X}$.
    \end{enumerate}
\end{proposition}

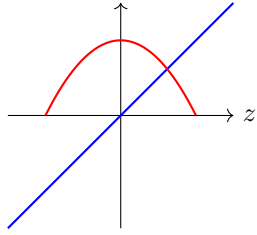
\begin{figure}
  \centering
  \begin{tikzpicture}[scale=1]
    \draw[->] (-1.5,0) -- (1.5,0) node[right] {$z$};
    \draw[->] (0,-1.5) -- (0,1.5);  

    \draw[blue, thick, domain=-1.5:1.5, samples=200]
      plot (\x,\x);

    \draw[red, thick, domain=-1:1, samples=200]
      plot (\x,{1 - \x*\x});
  \end{tikzpicture}
  \caption{Uniqueness of the optimal transport map for the quadratic leaky Monge problem (see proof of Proposition~\ref{the:solution_quadratic_leaky_monge}).}
  \label{fig:solution_quadratic_leaky_monge}
\end{figure}

    \begin{proof}[Proof of Proposition~\ref{the:solution_quadratic_leaky_monge}]
    The first part of the proof follows the same steps as the proof of Proposition~\ref{the:existence_optimal_maps}, we do not repeat them. Recall that the functions $g_y(x) = \frac{c(x,y)- \varphi(x)}{m(x,y)}$ and $f_x(y) = c(x,y) -  \psi(y) m(x,y)$ are Lipschitz and therefore Lebesgue a.e.~differentiable.

    Let us turn to the verification of (i). The characteristic equation for the optimal transport map is given for~$(x,y) \in \supp \pi^*$ by
    \begin{align}\label{equ:characteristic_equation_leaky_monge}
        \frac{(x-y) - \nabla_x \varphi (x)}{1- \frac{k}{2} |x-y|^2} + \frac{\frac{1}{2} |x-y|^2 - \varphi (x)}{\left( 1- \frac{k}{2}|x-y|^2\right)^2} \ k (x-y) = 0 \qquad \mbox{$\mu$-a.e.~$x$.}
    \end{align}
    Rearranging terms yields
    \begin{align}\label{equ:characteristic_equation_optimal_map_leaky_monge}
        (x-y) \left(1 - k \varphi(x) \right) =  \left( 1- \frac{k}{2}|x-y|^2 \right) \nabla_x \varphi(x).
    \end{align}
    We now argue that after keeping~$x$ fixed the equation~\eqref{equ:characteristic_equation_optimal_map_leaky_monge} uniquely characterizes~$y$. By~\eqref{equ:varphi_upper_bound}  we observe that~$\varphi(x) \leq \sup\limits_{x \in \mathcal{X} , y \in \mathcal{Y}} \frac{1}{2} |x-y|^2$. Combining this with our assumption
    \begin{align*}
            \inf_{x \in \mathcal{X} , y \in \mathcal{Y} } m(x,y) = 1-\frac{k}{2}\sup_{x\in \mathcal{X},y \in \mathcal{Y}}|x-y|^2 >0 ,
    \end{align*}
    it follows that 
    \begin{align}\label{equ:crucial_ingredient_non_perturbative}
        1 - k \varphi(x) >0.
    \end{align}
    In the case of~$\nabla_x \varphi_x =0$ it follows from~\eqref{equ:characteristic_equation_optimal_map_leaky_monge} that~$x-y=0$. Let us consider the case~$\nabla_x \varphi(x) \neq 0$. An Ansatz~$x-y = z \ \nabla_x \varphi(x)$ yields the equation
    \begin{align}
        \underbrace{1- \frac{k}{2}  |\nabla_x \varphi(x)|^2 z^2}_{ = 1- \frac{k}{2}|x-y|^2 = m(x,y) > 0} = (1 - k \varphi(x) ) z.
    \end{align}
    Let us interpret both sides as functions of $z$: The left-hand side is a concave quadratic function (as $- \frac{k}{2}  |\nabla_x \varphi(x)|^2 < 0$), which is restricted to the area where it attains positive values (as $m(x, y) > 0$). The right-hand side is a linear function with a zero intercept. Consequently, there exists a unique intersection, see Figure~\ref{fig:solution_quadratic_leaky_monge}. 
    Hence,~$y$ is uniquely characterized by~$x$ and an application of~Proposition~\ref{the:characterization_optimal_maps} yields the desired statement. \smallskip
    
    Let us turn to the verification of~(ii). The characteristic equation for the dual optimal transport map is given for $(x,y)\in \supp \pi^*$ by
    \begin{align}\label{equ:characteristic_equation_leaky_monge_dual}
        (y-x) + \psi(y) \  k \ (y-x) - \nabla_y \psi(y) \ \left( 1- \frac{k}{2}|x-y|^2 \right) = 0 \qquad \mbox{$\nu$-a.e.~$y$.}
    \end{align}
    Rearranging the equation gives
   \begin{align}
   \nabla_y \psi(y) \ \left( 1- \frac{k}{2}|x-y|^2 \right)  = (1 + k \psi(y) ) (y-x). 
    \end{align}
    Observe that, in light of~\eqref{equ:realtion_phi_psi} and Proposition~\ref{the:characterization_support_optimal_transport_plan}, for $(x,y) \in \supp \pi^*$ it a.s.~holds $\psi(y) = \frac{c(x,y) - \varphi(x)}{m(x,y)}$. This yields
    \begin{align*}
        1+ k \psi(y)  = 1 + k \frac{c(x,y) - \varphi(x)}{m(x,y)} & = \frac{m(x,y) + k c(x,y) - k\varphi(x)  }{m(x,y)}  = \frac{1- k \varphi(x)}{m(x,y)} \overset{\eqref{equ:crucial_ingredient_non_perturbative}}{>}0.
    \end{align*}
    Using the Ansatz~$y-x = z \ \nabla_y \psi(y)$ gives the equation
    \begin{align}
        1- \frac{k}{2}  |\nabla_y \psi(y)|^2 z^2 = (1 + k\psi(y)) z.
    \end{align}
   Observing that~$1- \frac{k}{2}  |\nabla_y \psi(y)|^2 z^2 = 1- \frac{k}{2}|x-y|^2= m(x,y) > 0$ we can conclude with the same argument as for (i).
   \end{proof}

\subsection{Dynamic formulation of non-conservative optimal transport}\label{sec:dynamic_formulation}
In the dynamic formulation, the transport of mass is modeled along a trajectory interpolating continuously between both marginal constraints. It has the advantage of allowing for finer modeling and interpretation of the mass transport. In our example of optimal re-balancing, the dynamic formulation would allow to calculate specific trading rates allowing for finer modeling for example by including dynamic price impact models. For simplicity, let us consider the cost functions if the form~$c(x,y)= |y-x|^a$ for~$a\geq 1$. 

\begin{definition}[Lagrangian formulation of~\eqref{equ:katorovich_problem}] \label{pro:lagrangian_NCOT}
  We consider the compact space~$\mathcal{X} \subset \mathbb{R}^n$ with a cost function~$c: \mathcal{X} \times \mathcal{X} \to \mathbb{R}$, $c(x,y)= |y-x|^a$ for some~$a\geq 1$, and a positive mass-change factor~$m: \mathcal{X} \times \mathcal{X} \to (0, \infty)$. Given two probability measures~$\mu, \nu \in \mathcal{P}(\mathcal{X})$, that are absolutely continuous w.r.t.~the Lebesgue measure, we define the set~$\mathcal{V}_m(\mu, \nu)$ of admissible vector fields~$v: [0,1] \times \mathcal{X} \to \mathbb{R}^n$ in the following way:\smallskip

  Let~$X(t, x)$ be the flow map given by the solution of the ODE
  \begin{align}\label{equ:lagrangian_coordinates}
      \partial_t X(t, x) = v(t, X(t, x)) \qquad \mbox{and} \qquad X(0,x)=x.
  \end{align}
  Then~$v \in \mathcal{V}_{m}(\mu, \nu)$ if and only if
  \begin{align}\label{equ:admissible_vector_field}
      \left( X(1, \cdot) \right)_{\sharp} \left( \frac{m(\cdot,X(1,\cdot))}{\int_\mathcal{X} m(x ,X(1,x)) \mu(dx)} \mu \right) =  \ \nu.
  \end{align}
    Then the dynamic formulation of~\eqref{equ:katorovich_problem} is given by
    \begin{align}\label{equ:generalized_benamou_brenier0} \tag{GBB}
        \inf_{v \in \mathcal{V}_m(\mu, \nu)} \int_{0}^1 \int_{\mathcal{X}} |v(t, X(t)|^a \mu(dx) dt.
    \end{align}
\end{definition}

\begin{remark}\label{rem:eulerian_formulation}
    The equations~\eqref{equ:lagrangian_coordinates} define Lagrangian coordinates of the transport of mass along the vector field~$v$. The equivalent Eulerian formulation~$p:[0,1] \times \mathcal{X} \to \mathbb{R}$ is given by the classical continuity equation 
    \begin{align}\label{equ:continuity_equation}
        \partial_t p = \nabla \cdot (v p) \qquad \mbox{and} \qquad p(t=0)= \mu.
    \end{align}
    Therefore,~\eqref{equ:generalized_benamou_brenier0} is a generalization of the classical Benamou-Brenier formulation~\cite{benamou2000computational} and recovers it in the special case~$m(x,y)=1$.
\end{remark}

    Let us briefly discuss the well-posedness of the dynamic optimization problem~\eqref{equ:generalized_benamou_brenier0}. One would require a-priori regularity assumptions on the vector field~$v(t,x)$ to guarantee that Lagrangian coordinates~\eqref{equ:lagrangian_coordinates} exist. However, we did not state those regularity assumptions on purpose. The proof of Proposition~\ref{the:equivalence_dynamic_static} below shows that the existence of an optimal transport map~$T: \mathcal{X} \to \mathcal{X}$ is sufficient to guarantee the well-posedness of~\eqref{equ:generalized_benamou_brenier0}. This leverages one of the main advantage of the relative simple-minded formulation of~\eqref{equ:katorovich_problem} compared to other variants of optimal transport that allow for mass creation or destruction. The advantage is that~\eqref{equ:katorovich_problem} allows to stay close to the classical theory of optimal transport. Here we use the fact that the existence of optimal transport maps is relative easy to deduce for~\eqref{equ:katorovich_problem}, simplifying our arguments. However, let us remark that that the existence of optimal transport maps is not needed in the classical setting for deducing the dynamic formulation (see e.g.~Chapter~8 in~\cite{ambrosio2008gradient}). We expect that this carries over to the non-conservative setting of~\eqref{equ:katorovich_problem}.
    The next statement shows the equivalence of the static problem~\eqref{equ:katorovich_problem} and the dynamic problem~\eqref{equ:generalized_benamou_brenier0}.
\begin{proposition}\label{the:equivalence_dynamic_static}
    Let us consider the setting of Definition~\ref{pro:lagrangian_NCOT} and assume that an optimal transport map~$T: \mathcal{X} \to \mathcal{X}$ exists for~\eqref{equ:katorovich_problem}. Then it holds   
    \begin{equation}\label{equ:generalized_benamou_brenier} 
        \inf_{v \in \mathcal{V}_m(\mu, \nu)}  \int_{0}^1 \int_{\mathcal{X}} |v(t, X(t,x))|^a \mu(dx) dt = \inf_{\pi \in \Gamma_m(\mu,\nu)} \int_{\mathcal{X}\times \mathcal{X}} |y-x|^{a} \pi(dx,dy) .
    \end{equation}   
\end{proposition}
\begin{proof}
    We adapt the standard Benamou-Brenier~\cite{benamou2000computational} argument. In the first step, we show that for arbitrary $v \in \mathcal{V}_m(\mu, \nu)$ it holds
    \begin{align}\label{equ:dynamic_lerger_static}
        \int_{0}^1 \int_{\mathcal{X}} |v(t, X(t,x))|^a \mu(dx) dt \geq \int_{\mathcal{X}} |T(x) -x|^a\mu(dx). 
    \end{align}
    Indeed, using Jensen's inequality and the Lagrangian coordinates~\eqref{equ:lagrangian_coordinates} gives
    \begin{align*}
        \int_{0}^1 \int_{\mathcal{X}} |v(t, X(t,x))|^a \mu(dx) dt & \geq \int_{\mathcal{X}} \left| \int_{0}^1 v(t, X(t,x)) dt \right|^a \mu(dx)  \\
        & = \int_{\mathcal{X}} \left| \int_{0}^1  \partial_t X(t, x ) dt  \right|^a \mu(dx) \\
        & = \int_{\mathcal{X}} \left| X(1,x) - x   \right|^a \mu(dx) \\
        & \geq \int_{\mathcal{X}} \left| T(x) - x   \right|^a \mu(dx),
    \end{align*}
    where in the last step we used that the map~$\mathcal{X} \ni x \mapsto X(1,x) \in \mathcal{X}$ defines an admissible transport map by~\eqref{equ:admissible_vector_field}, and~$T$ is the optimal one. \smallskip

    In the second step, we show that 
    \begin{align*}
       \inf_{v \in \mathcal{V}_m(\mu, \nu)} \int_{0}^1 \int_{\mathcal{X}} |v(t, X(t,x))|^a \mu(dx) dt \leq \int_{\mathcal{X}} \left| T(x) - x   \right|^a \mu(dx)
    \end{align*}
    by defining an admissible vector field~$\bar v$ that achieves the right hand side. We follow the standard Benamou-Brenier argument and define the Lagrangian coordinates
    \begin{align*}
        X(t,x) = x + t (T(x) -x)  \qquad \mbox{with} \qquad  \bar v(t,X(t)) := \partial_t X(t,x) = T(x)-x.
    \end{align*}
    As the optimal transport plan $T$ is admissible, the associated vector field~$\bar v$ is also admissible, i.e.~it satisfies~\eqref{equ:admissible_vector_field}. As the vector field~$\bar v$ does not depend on time we get 
    \begin{align*}
         \int_{0}^1 \int_{\mathcal{X}} |\bar v(t, X(t, x)|^a \mu(dx) dt & =  \int_{0}^1 \int_{\mathcal{X}} |T(x)-x|^a \mu(dx) dt \\
            & = \int_{\mathcal{X}} |T(x)-x|^a \mu(dx),
    \end{align*}
    which concludes the argument.
\end{proof}

We now outline the Eulerian description of the optimization Problem~\ref{pro:lagrangian_NCOT}, which will give insight into the potential governing the creation or destruction of mass.  Instead of tracking individual particles through the flow map $X(t,\cdot)$, we define the time–dependent density $\varrho(t,\cdot)$ that results from pushing forward the initial mass $\mu$ while weighting each trajectory by the mass-change factor $m$. It represents the physical density at time~$t$ taking creation and destruction of mass into account. The density~$\varrho(t, \cdot)$ interpolates for~$0 \leq t \leq 1$ between $\mu$ and the measure $\nu$ re-normalized to account for the change of mass during the transport. Compatibility with the Lagrangian description enforces that, $\varrho(t,\cdot)$ is the density of
\begin{align}\label{equ:density_varrho}
    \left( X(t, \cdot) \right)_{\sharp} \left(m(\cdot,X(t,\cdot)) \mu \right).    
\end{align}

In particular, for any smooth test function $g$, which we take to be compactly supported in $\mathcal X$,
\begin{align*}
\int_\mathcal X g(x)\varrho(t,x)dx &= \int_\mathcal X g(x) \ d\left(\left( X(t, \cdot) \right)_{\sharp} \left(m(\cdot,X(t,\cdot)) \mu \right)\right)\\
&= \int_\mathcal X g(X(t,u)) m(u, X(t,u)) \mu(du)
\end{align*}
Then, we can compute derivatives to see that,
\begin{align*}
    & \int_\mathcal X g(x)\partial_t\varrho(t,x)dx  = \partial_t \int_\mathcal X g(x)\varrho(t,x)dx  \\
    &= \partial_t \left( \int_\mathcal X g(X(t,u)) m(u, X(t,u)) \mu(du) \right) \\
    &= \int_\mathcal X \left( \ \nabla g(X(t,u)) m(u, X(t,u)) + g(X(t,u))\nabla_x m(u, X(t,u)) \ \right) \cdot v(t,X(t,u)) \mu(du) \\
    &=\int_{\mathcal{X}} \nabla g(X(t,u)) \cdot v(t,X(t,u))  \ m(u, X(t,u)) \mu (du) \\&\quad\quad + \int_\mathcal X g(X(t,u)) \nabla_X m(u,X(t,u))  \cdot v(t, X(t,u)) \ \frac{m(u,X(t,u))}{m(u,X(t,u))} \ \mu (du) \\
    &=\int_\mathcal X \nabla g(x) \cdot v(t, x) \varrho(t, x) \ dx + \int_\mathcal X \frac{g(x)}{m(X^{-1}(t,x),x)}\nabla_x m(X^{-1}(t,x),x)\cdot v(t,x) \varrho(t,x) \ dx\\
    &= -\int_\mathcal X g(x) \nabla \cdot (v(t,x) \varrho(t,x)) \ dx + \int_\mathcal X g(x) \nabla_x \ln m(X^{-1}(t,x), x) \cdot v(t, x) \varrho(t, x) \ dx.
\end{align*}
The last calculation leads to the following observation.
\begin{remark}[Eulerian formulation of~\eqref{equ:katorovich_problem}]
    The density $\varrho$ defined via~\eqref{equ:density_varrho} satisfies the following PDE,
\[
\partial_t \varrho(t, x) + \nabla \cdot (v(t,x)\varrho(t,x)) = \sigma(t,x),
\]
where 
\begin{align*}
    \sigma(t,x) := \varrho(t,x)\nabla_x \log  m(X^{-1}_t(x), x) \cdot v(t,x).
\end{align*}
In general, the densities $p$ (see Remark~\ref{rem:eulerian_formulation}) and $\varrho$ are related via
\[
p(t,x) = \frac{\varrho(t,x)}{m(X^{-1}(t,x),x)}.
\]
In the special case of quadratic cost, $c(x,y) = \frac12 |x-y|^2$, and $m(x,y) = e^{c(x,y)}$, it holds $$\sigma(x,t) = \varrho(x,t)v(x,t)\cdot \left( x - X^{-1}(t,x)\right).$$
\end{remark}

\section{Relation to existing optimal–transport variants}
\label{sec:literature}

In this section we compare the framework of this article to existing variants of optimal transport which allow flexibility in the marginal constraints. As the literature on those variants is abundant, we do not claim that this overview is exhaustive. We focus on three popular approaches,  unbalanced optimal transport in Section~\ref{sec:unbalanced_optimal_transport}, entropic optimal transport in Section~\ref{sec:entropic_optimal_transport}, and unnormalized optimal transport in Section~\ref{sec:unnormalized_optimal_transport}. \\

\subsection{Relation to unbalanced optimal transport}~\label{sec:unbalanced_optimal_transport}

Unbalanced optimal transport became an umbrella term for optimal transport variants that weaken the original marginal conditions of a transport plan. However, here we use the term unbalanced optimal transport to refer to the framework originally developed in~\cite{kondratyev2016new, chizat2018interpolating} and extended in~\cite{chizat2018unbalanced}. The starting point is the dynamic formulation of optimal transport via the Benamou-Brenier formula and allowing for mass creation or destruction along the trajectories of mass transport. In~\cite{chizat2018unbalanced}, it was shown that the dynamic formulation is equivalent to a static optimization problem described below. \\

Let $\mathcal{X},\mathcal{Y} \subset \R^n$ denote two compact sets. In unbalanced optimal transport, the cost function~$ \bar c= \bar c(x_0,m_0,x_1,m_1)$ determines the cost of transporting/transforming a Dirac mass of size~$m_0$ located at~$x_0$ into a Dirac mass of size~$m_1$ located at~$x_1$. In this context, the cost function 
\[
\bar c \colon\; (\mathcal{X}\times[0,\infty)) \times (\mathcal{Y}\times[0,\infty))  \;\to\; [0,\infty]
\]
is assumed to be l.s.c. in all its arguments and jointly sublinear in $(m_0, m_1)$, i.e.~positively 1-homogeneous and subadditive. 

In classical optimal transport, the decision space consists of couplings, which preserve mass. To allow for creation or destruction of mass, the unbalanced optimal transport introduces the notion of semi-couplings $(\pi_0, \pi_1)$, namely a pair of two couplings. The interpretation is that $\pi_0(x,y)$ represents the amount of mass taken from $\mu$ at point $x$ and then transported to possibly different amount of mass $\pi_1(x,y)$ at point $y$ of $\nu$.

\begin{definition}[Semi-couplings]
Given two marginals $\mu \in \mathcal{M}_+(\mathcal{X})$ and~$\nu \in \mathcal{M}_+(\mathcal{Y})$, the set of semi-couplings is
	\begin{align}
		\bar \Gamma(\mu,\nu) \eqdef \left\{
			(\pi_0,\pi_1) \in \mathcal{M}_+(\mathcal{X} \times \mathcal{Y})^2  \colon
			(\mbox{Proj}_0) \pi_0 = \mu,\, (\mbox{Proj}_1) \pi_1 = \nu
			\right\}.
	\end{align}
\end{definition}
This leads to the static Kantorovich formulation of unbalanced optimal transport.

\begin{definition}[Unbalanced Kantorovich problem]
	The static unbalanced optimal transport problem is given by
	\begin{align}\tag{UKP}\label{equ:kantorovich_uot}
 		\inf_{(\pi_0, \pi_1) \in \bar \Gamma(\mu, \nu)}  \iint_{\mathcal{X} \times \mathcal{Y}} \bar c \left( x,\frac{ \pi_0}{ \gamma},y,\frac{ \pi_1}{ \gamma}\right)  \ d \gamma(x,y) , 
	\end{align}
where $\gamma \in \mathcal{M}_+(\mathcal{X} \times \mathcal{Y})$ is any measure such that $\pi_0, \pi_1 \ll \gamma$. The integral is well-defined since $c$ is jointly 1-homogeneous w.r.t.\ the mass variables. 
\end{definition}
The interpretation of~\eqref{equ:kantorovich_uot} is that one looks for the most efficient way to transform the measure~$\mu$ into the measure~$\nu$ by a combination transporting mass that is modulated by creation or destruction along its way. In~\cite[Proposition 3.4]{chizat2018unbalanced} it was shown that~\eqref{equ:kantorovich_uot} allows a minimizer. The dual problem of~\eqref{equ:kantorovich_uot} is quite subtle and the existence of a solution of the dual problem is not easily guaranteed. The same is true for strong duality.  \\

The main difference between~\eqref{equ:kantorovich_uot} and~\eqref{equ:katorovich_problem} is the following. In~\eqref{equ:kantorovich_uot} the mass difference between~$\mu$ and~$\nu$ is fixed. In~\eqref{equ:katorovich_problem}, only the shape of the second marginal~$\nu$ is prescribed allowing for finding an optimal transport among couplings with varying mass differences. There is the following link between both problems.
\begin{proposition}[Link between \eqref{equ:katorovich_problem} and~\eqref{equ:kantorovich_uot}]\label{the:ncot_vs_uot}
    Let~$c$ and~$m$ be the cost and mass-change factor of~\eqref{equ:katorovich_problem}. Then the function
    \begin{align}
        \bar c(x,y,m_0, m_1) = \begin{cases}
            c(x,y), & \mbox{ if } m_1= m(x,y) m_0, \\
            \infty, & \mbox{ else,}
        \end{cases}
    \end{align}
    is a cost function $\bar c \colon\; (\mathcal{X}\times[0,\infty)) \times (\mathcal{Y}\times[0,\infty))  \;\to\; [0,\infty]$ and 
    \begin{align}\label{equ:ncot_vs_uot_1}
        \inf_{\Gamma_m(\mu, \nu)} &  \iint c(x,y) \pi(dx, dy) = \inf_{Z >0} \ \inf_{\bar \Gamma(\mu, Z \nu)}  \iint_{\mathcal{X} \times \mathcal{Y}} \bar c \left( x,\frac{ \pi_0}{ \gamma},y,\frac{ \pi_1}{ \gamma}\right)  \ d \gamma(x,y).
    \end{align}
    In particular, if $Z = \iint m(x,y) \pi^*(dx, dy)$ for the optimal~$\pi^*$ of~$\eqref{equ:katorovich_problem}$, then  
    \begin{align}\label{equ:ncot_vs_uot_2}
        \inf_{\pi \in \Gamma_m(\mu, \nu)} &  \iint c(x,y) \pi(dx, dy) = \inf_{\bar \Gamma(\mu, Z \nu)}  \iint_{\mathcal{X} \times \mathcal{Y}} \bar c \left( x,\frac{ \pi_0}{ \gamma},y,\frac{ \pi_1}{ \gamma}\right)  \ d \gamma(x,y).
    \end{align}    
\end{proposition}
The last proposition states that one can recover the solution~\eqref{equ:katorovich_problem} via minimizing over a family of~\eqref{equ:kantorovich_uot} problems. If one also knows how much mass is lost or generated in the optimal~\eqref{equ:katorovich_problem} transport then the problems~\eqref{equ:katorovich_problem} and~\eqref{equ:kantorovich_uot} are equivalent. However, this is not known apriori. The verification of the last proposition follows directly from careful examination of the definitions and is left as an exercise. The simpler structure of~\eqref{equ:katorovich_problem} compared to~\eqref{equ:kantorovich_uot} allows for easier arguments when solving the dual problem, deriving strong duality, or studying optimal transport maps. However, the framework~\eqref{equ:kantorovich_uot} allows to model more general situations.

\subsection{Relation to entropic optimal transport in the sense of~\cite{liero2018optimal}}\label{sec:entropic_optimal_transport}
The entropic transport problem, introduced in~\cite{liero2018optimal}, relaxes classical mass-conserving optimal transport by allowing creation and annihilation of mass, captured via entropic penalization. Entropic transport has been a very active area and an overview of recent developments can be found e.g.~in~\cite{TreZua23}. 
For the discussion, let us briefly introduce the main concepts of entropic optimal transport in a slightly simplified setting. 
\begin{definition}[Entropy functions and relative entropies]
    A function~$F: [0, \infty) \to [0, \infty]$ is called \emph{admissible entropy function} if it is convex, lower semicontinous and there is a point~$0 < x < \infty$ such that~$F(x)< \infty$. Given two finite measures~$\gamma_1$ and~$\gamma_2$ on a space~$\mathcal{X}$, the relative entropy~$\mathcal{F}(\gamma_1|\gamma_2)$ w.r.t.~entropy function~$F$ is defined as
    \begin{align}
        \mathcal{F}(\gamma_1|\gamma_2) = \begin{cases}
            \int F(\frac{d\gamma_1}{d \gamma_2}) \ d \gamma_2, &\mbox{if } \gamma_1 \ll \gamma_2, \\
            \infty, & \mbox{else}.
        \end{cases}
    \end{align}
\end{definition}
\begin{definition}[Entropic Transport Problem]\label{def:entropy_transport}
    We consider two measures~$\mu \in \mathcal{X}$ and~$\nu \in \mathcal{Y}$ and relative entropies~$\mathcal{F}_1$ and~$\mathcal{F}_2$ on $\mathcal{X}$ and $\mathcal{Y}$, respectively. Given a lower semicontinuous cost function~$c: \mathcal{X} \times \mathcal{Y} \to \mathbb{R}$  we optimize the expression
    \begin{align}\tag{ETP}
        \inf_{\pi \in \mathcal{M} (\mathcal{X} \times \mathcal{Y})}  \int_{\mathcal{X} \times \mathcal{Y}} c(x,y) \pi(dx, dy) + \mathcal{F}_1((\mbox{Proj}_1)_\sharp \pi|\mu) +  \mathcal{F}_2((\mbox{Proj}_2)_{\sharp} \pi|\nu) .
    \end{align}
\end{definition}
\begin{remark} Entropic transport contains the classical optimal transport problem as a degenerate case, namely choosing the admissible entropy function as 
\begin{align}\label{equ:entropy_function_classical_OT}
    F(z) = \begin{cases}
        0, & \mbox{if } z =1\\
        \infty, & \mbox{else}.
    \end{cases}
\end{align}    
\end{remark}

\begin{remark}
  The entropy functionals in Definition~\ref{def:entropy_transport} are applied to the marginals of the coupling~$\pi$. In the optimization problem~\eqref{equ:katorovich_problem}, the second marginal of the transport is not determined a-priori. So the problem of Definition~\ref{def:entropy_transport} does not contain~\eqref{equ:katorovich_problem} unless it is generalized to the following degenerate problem. Given a measure~$\pi \in \mathcal{M} (\mathcal{X} \times \mathcal{Y})$ let us consider the transformation~$T (\pi) \in \mathcal{Y}$ defined by 
  \begin{align}
      T(\pi) = \frac{\int_{\mathcal{X}} m(x,y) \pi(dx,y)}{\iint_{\mathcal{X} \times \mathcal{Y}} m(x,y) \pi(dx,dy)}.
  \end{align}
  Then~\eqref{equ:katorovich_problem} becomes the following generalized entropic transport problem.
    \begin{align*}
        \inf_{\pi \in \mathcal{M} (\mathcal{X} \times \mathcal{Y})}  \int_{\mathcal{X} \times \mathcal{Y}} c(x,y) \pi(dx, dy) + \mathcal{F}((\mbox{Proj}_1)_\sharp \pi|\mu) +  \mathcal{F}(T(\pi)|\nu) .
    \end{align*}  
\end{remark}

\subsection{Relation to unnormalized optimal transport}\label{sec:unnormalized_optimal_transport}

Unnormalized optimal transport~\cite{gangbo2019unnormalized} uses a computational fluid mechanics approach to allow the transport of two measures with different mass. A spatially homogeneous term is added to the continuity equation to account for the creation or loss for mass. Additionally, a regularization term is added the minimization of the kinetic energy leading to a generalization of the Benamou-Brenier formulation of optimal transport.

\begin{definition}[Unnormalized Optimal Transport]
    Let us consider a bounded domain~$\mathcal{X} \subset \mathbb{R}^n$, two nonnegative measures ~$\mu, \nu \in \mathcal{M}(X)$, and parameters~$p \geq 1$ and~$\alpha >0$. The unnormalized optimal transport problem is given by
    \begin{equation}\label{UBB}
        \inf_{v,f} \int_0^1 \int_{\mathcal{X}} \| v(t,x)\|^p \varrho(t,x) dx dt + \frac{1}{\alpha} \int_0^1 |f(t)|^p dt, 
    \end{equation}
    such that the vector field~$v$ satisfies the zero flux condition and~$\rho(t)$ satisfies the dynamical constraint
    \begin{align*}
        \partial_t \varrho  + \nabla \cdot ( v \varrho) = f(t) \qquad \mbox{and} \qquad \varrho(0)= \mu, \quad \varrho(1)= \nu.
    \end{align*}
\end{definition}
The regularization term~$\int_0^1 |f(t)|^p dt$ constitutes the main difference to the other variants of optimal transport. The dynamic formulation~\eqref{equ:generalized_benamou_brenier0} of~\eqref{equ:katorovich_problem} does not produce such a regularization term (see Section~\ref{sec:dynamic_formulation}) and therefore the two optimal transport variants are distinct. 

\section*{Acknowledgment}
We want to express our deepest gratitude to Wilfrid Gangbo for his interaction. His insights and knowledge immensely helped to write this article. We also want to thank Lenaic Chizat, Pablo L\'opez Rivera, Bernhard Schmitzer, and Moritz Voss for their valuable comments.

\appendix
\section{Proofs of results from Section~\ref{sec:rebalancing_general_graphs}}
\label{sec:existenceReb}

\subsection{Existence of consistent price vector} \label{append:consistent-price}
Here we prove that existence of consistent price measure is equivalent to no arbitrage.
\begin{proof}[Proof of Proposition~\ref{prop:consistent-price}]
To prove that no arbitrage is  \emph{necessary}, assume that a consistent price vector $\bq$ exists. Let $\Gamma = ((v_1, v_2), (v_2, v_3), \dots, (v_M, v_1))$ be an arbitrary directed cycle in the graph. Then, by definition of the  consistent price vector, it holds
\begin{align*}
\prod_{e \in \Gamma} P_e \leq \prod_{e \in \Gamma} Q_e = \frac{q_{v_1}}{q_{v_2}} \cdot \frac{q_{v_2}}{q_{v_3}} \cdot \dots \cdot \frac{q_{v_M}}{q_{v_1}}   = 1.  
\end{align*} 

Now let us show that no arbitrage is \emph{sufficient}. Assume that the market $\mathcal{G} = (V, E, P)$ satisfies assumptions A1 and A2 and  define graph $\mathcal{H} = (V, E, w)$ by setting
\[
w_{(i,j)} := - \log P_{(i,j)} \quad \quad \text{ for all } (i, j) \in E.
\]
Since $\mathcal{G}$ satisfies no arbitrage (A2), for any directed cycle $\Gamma$ (in $\mathcal{G}$ or $\mathcal{H}$) it holds
\[
\prod_{e \in \Gamma} P_e \leq 1, \quad \quad \text{ or equivalently, } \quad \quad \sum_{e \in \Gamma} w_e \geq 0.
\]
Therefore, graph $\mathcal{H}$ contains no cycles of negative weight. For $i \in V$ define 
\[
\phi_i := \text{ length of the shortest directed path from 1 to } i \text{ in graph } \mathcal{H}.
\]
Since $\mathcal{H}$ contains no negative cycles, shortest paths exist and $\phi_1, \dots, \phi_N$ are well-defined, see~\cite[Theorem 24.4]{CormenETAL09}. Now consider any edge $(i,j) \in E$. Since $\phi_j$ is the length of the shortest path, it holds
\[
\phi_j \leq \phi_i + w_{(i,j)}, \quad \quad \text{ or equivalently, } \quad \quad \log P_{(i,j)} = -w_{(i,j)} \leq \phi_i - \phi_j.
\]
By setting $q_i = e^{\phi_i}$ for $i=1, \dots, N$ we obtain a consistent price vector.
\end{proof}

\subsection{Existence of optimal rebalancing} \label{append:existenceRebalancing}
As the first step towards proving Theorem~\ref{thm:OptRebalancing}, let us observe that the rebalancing problem is a linear program with an objective bounded from below.
\begin{lemma} \label{lemma:RebalancingProb}
For an admissible trade $\xi \in \Xi$ it holds $C(\xi) \geq 0$.
Moreover, the rebalancing Problem~\ref{prob:rebalance} is (equivalent to) a linear program. 
\end{lemma}
\begin{proof}
Recall that $C(\xi)$ can be expressed via~\eqref{eq:costC}.  
Admissible trade $\xi \in \Xi$ satisfying $\xi_{(i,j)} \geq 0$, consistency $P_{(i,j)} \leq \frac{q_i}{q_j}$ of $\bq$ and non-negative prices $P$ imply $C(\xi) \geq 0$.

To simplify restating the target proportion constraints, let us introduce an additional variable $\eta := C(\xi)$. The rebalancing problem can be reformulated as
\begin{align*}
\underset{\xi \in \Xi, \eta \in \R}{\text{minimize }} \quad  & \eta \\
\text{subject to } \quad &\eta = \sum\limits_{i=1}^N q_i \sum\limits_{j=1}^N  (\xi_{(i, j)} - P_{(j,i)} \xi_{(j,i)}), \\
&x_i + \sum\limits_{j=1}^N  (P_{(j,i)} \xi_{(j,i)} - \xi_{(i, j)}) \geq 0, \quad \quad i \in V, \label{prob:eq2} \\
&\nu_i \left( \sum\limits_{k=1}^N q_k  x_k - \eta \right) = q_i \left( x_i + \sum\limits_{j=1}^N  (P_{(j,i)} \xi_{(j,i)} - \xi_{(i, j)}) \right), \quad \quad i \in V.
\end{align*}
Since the objective and all constraints are linear in $\xi$ and $\eta$ jointly, this is a linear program.
\end{proof}

As the next step towards proving Theorem~\ref{thm:OptRebalancing}, we consider the special case of a \emph{star-shaped market}. We use the terminology and notation introduced in Examples~\ref{ex:star1}, \ref{ex:star2} and~\ref{ex:star3}.

\begin{example}\label{ex:star4}
Consider the \emph{star-shaped market}. Given a current portfolio $\bx = (x_1, x_2, \dots, x_N)$, we aim to show existence of a \emph{feasible} rebalancing trade yielding the desired target proportions $\nu$. Note  we do not aim for the optimal trade here, instead we use a (sub-optimal) strategy consisting of two steps:
\begin{enumerate}
    \item First, trade all available units of assets $i=2, \dots, N$ for the num\'eraire via $\xi_{(i,1)} := x_i$, creating a 'money pot' portfolio
    \[
    \left(x_1 + \sum_{i=2}^N P^b_i x_i, \, 0, \,  \dots, \, 0 \right).
    \]
    \item Second, use the 'money pot' held in the num\'eraire only to buy a portfolio with the desired target proportions. Proposition~\ref{prop:star} below shows that such a trade exists.
\end{enumerate}
Note that a trade constructed in this way will not satisfy~\eqref{eq:star1}. 
\end{example}

\begin{proposition}\label{prop:star}
Consider the \emph{star-shaped market} as describe in Examples~\ref{ex:star1}, \ref{ex:star2} and~\ref{ex:star3}. For a portfolio $\bx = (1, 0, \dots, 0)$ there exists a feasible rebalancing trade $\xi \in \Xi(\nu)$ achieving the desired target $\nu$.
\end{proposition}
\begin{proof}
Firstly, set
\[
\xi_{(i,j)} = 0 \quad \text{ if } i \neq 1 \text{ or } j=1. 
\]
Consequently, the trade $\xi$ is fully determined by the quantities $(\xi_{(1,i)})_{i=2, \dots, N}$. For the sake of a simpler notation in the reminder of the proof we denote by $\bz := \bx + \Delta \bx (\xi)$ the portfolio position after the trade. Recall that it holds
\begin{align} \label{eq:prop-star1}
    z_i = P_{(1, i)} \xi_{(1, i)} = \frac{1}{P^a_i} \xi_{(1, i)} \quad \quad  i =2, \dots, N
\end{align}
and 
\begin{align} \label{eq:prop-star2}
    z_1 = x_1 - \sum_{i=2}^N \xi_{(1, i)} = 1 - \sum_{i=2}^N P^a_i z_i.
\end{align}
Note that \eqref{eq:prop-star2} represents the self-financing condition and admissible trade $\xi \in \Xi$ can be recovered from $\bz$ via \eqref{eq:prop-star1}.

The constraints of the feasible set $\Xi(\nu)$ expressed in terms of after-trade position $\bz$ are
\begin{align}
&z_i \geq 0 \quad &i = 1, \dots, N, \label{eq:prop-star3}\\
&q_i z_i + \nu_i \cdot \sum_{j=1}^N (P^a_j - q_j) z_j = \nu_i, \quad &i = 1, \dots, N, \label{eq:prop-star4}
\end{align}
where the self-financing condition \eqref{eq:prop-star2} was used to obtain \eqref{eq:prop-star4}.
Recall that $P^a_1 = q_1 = 1$ and $P^a_i \geq q_i$ for $i = 2, \dots, N$. In order to show existence of a vector $\bz$ solving the system \eqref{eq:prop-star3}-\eqref{eq:prop-star4} we use the Farkas' Lemma. Let us, therefore, introduce vectors  $\bq = (q_1, \dots, q_N)^T$ and $\bP_a = (P^a_1, \dots, P^a_N)^T$. In matrix notation, the system \eqref{eq:prop-star3}-\eqref{eq:prop-star4} is
\begin{align} \label{eq:prop-star5}
\left( \text{diag } \bq + \nu \cdot (\bP_a - \bq)^T \right) \cdot \bz = \nu, \quad \bz \geq 0.  
\end{align}

According to Farkas' Lemma, solution of \eqref{eq:prop-star5} exists if and only if the complementary system
\begin{align} \label{eq:prop-star6}
\left( \text{diag } \bq + \nu \cdot (\bP_a - \bq)^T \right)^T \cdot \by \geq 0, \quad \nu^T \by < 0
\end{align}
has no solution. Consider \eqref{eq:prop-star6}, $\nu^T \by < 0$ implies that there exists $j \in \{1, \dots, N\}$ with $y_j < 0$. Looking at the the $j$-th row of $\left( \text{diag } \bq + \nu \cdot (\bP_a - \bq)^T \right)^T \cdot \by$, we obtain a contradiction
\[
\left[\left( \text{diag } \bq + \nu \cdot (\bP_a - \bq)^T \right)^T \cdot \by \right]_j =  q_j y_j + (P^a_j - q_j) \cdot \nu^T \cdot \by < 0
\]
since $P^a_j \geq  q_j > 0$.
Therefore, the complementary system \eqref{eq:prop-star6} does not allow a solution and, by Farkas' Lemma, a solution $\bz$ to system \eqref{eq:prop-star5} exists. Finally, note that any vector $\bz$ satisfying \eqref{eq:prop-star4} also satisfies the self-financing condition \eqref{eq:prop-star2}. Therefore, vector $\bz$ corresponds to a feasible trade $\xi \in \Xi(\nu)$ that can be recovered via \eqref{eq:prop-star1}. 
\end{proof}

To construct a feasible rebalancing trade on the general market, we will convert the market to a hypothetical start-shaped one. For this purpose, for each asset $i=2, \dots, N$ fix one directed path $\Gamma_{1 \to i}$ from vertex $1$ to vertex $i$ and one directed path $\Gamma_{i \to 1}$ from vertex $i$ to vertex $1$, respectively. We set
\[
P_{1 \to i} := \prod_{e \in \Gamma_{1 \to i}} P_e \quad \text{ and } \quad P_{i \to 1} := \prod_{e \in \Gamma_{i \to 1}} P_e.
\]
Such paths exist on connected graph and a path visits each edge at most once.

\begin{proposition}\label{prop:existence}
For a portfolio $\bx = (x_1, x_2, \dots, x_N) \geq 0$ on a (general) market $\mathcal{G}$ there exists a feasible rebalancing trade $\xi \in \Xi(\nu)$ achieving the desired target $\nu$.
\end{proposition}
\begin{proof}
In the first stage, we convert portfolio $\bx$ into a 'money pot' portfolio
\[
\tilde{\bx} = \left( x_1 + \sum_{i=2}^N  P_{i \to 1} x_i, \, 0, \dots, 0 \right)
\]
by trading along the fixed directed walks. Specifically, for each asset $i \in \{2, \dots, N\}$ iteratively define a trade $\xi^{i \to 1}$.  Let the fixed path $\Gamma_{i \to 1}$ consist of edges\footnote{Each walk $\Gamma$ mentioned in this proof can consist of different number $M = M_\Gamma$ of visited edges. We opt for keeping the notation as simple as possible.} $(e_1, e_2, \dots, e_M)$. Set $\xi^{i \to 1}_{e_1} := x_i$ and 
\[
\xi^{i \to 1}_{e_k} := P_{e_{k-1}} \cdot \xi^{i \to 1}_{e_{k-1}} \quad \quad k = 2, \dots, M,
\]
all other elements of $\xi^{i \to 1}$ are zero. Applying an admissible trade $\sum_{i=2}^N \xi^{i \to 1}$ to portfolio $\bx$ yields the 'money pot' portfolio $\tilde{\bx}$.

In the second stage, we convert the 'money pot' portfolio $\tilde{\bx}$ into a portfolio with desired target proportions $\nu$. For this purpose, we construct a hypothetical start-shaped market $\tilde{\mathcal{G}} = (V, \tilde{E}, \tilde{P})$ defined via
\[
\tilde{E} = \{ (1, i) : i = 2, \dots, N\} \cup \{ (i, 1) : i = 2, \dots, N\}
\]
and 
\[
\tilde{P}_{(1, i)} = P_{1 \to i} \quad \text{ and } \quad \tilde{P}_{(i,1)} = P_{i \to 1}.
\]
Note that we do not change the consistent price vector $\bq$. By Proposition~\ref{prop:star} and scaling, for the 'money pot' portfolio $\tilde{\bx}$ on the hypothetical star-shaped market $\tilde{\mathcal{G}}$ there exists a feasible (hypothetical) rebalancing trade $\tilde{\xi}$ achieving the target $\nu$. Note that only elements $\tilde{\xi}_{1 \to 2}, \dots, \tilde{\xi}_{1 \to N}$ are non-zero.

Now we convert $\tilde{\xi}_{1 \to 2}, \dots, \tilde{\xi}_{1 \to N}$ into admissible trades $\xi^{1 \to 2}, \dots, \xi^{1 \to N}$ on the market $\mathcal{G}$. Fix an asset $i \in \{2, \dots, N\}$ and let the fixed path $\Gamma_{1 \to i}$ consist of edges $(e_1, e_2, \dots, e_M)$. Set $\xi^{1 \to i}_{e_1} := \tilde{\xi}_{1 \to i}$ and 
\[
\xi^{1 \to i}_{e_k} :=  P_{e_{k-1}} \cdot \xi^{1 \to i}_{e_{k-1}} \quad \quad k = 2, \dots, M,
\]
all other elements of $\xi^{1 \to i}$ are set to be zero.

Finally, observe that applying the trade $\sum_{i=2}^N \xi^{1 \to i}$ (along the graph $\mathcal{G}$) to $\tilde{\bx}$ yields the same post-trade portfolio as applying the hypothetical trade $\tilde{\xi}$ (along $\tilde{\mathcal{G}}$). Therefore, $\xi = \sum_{i = 2}^N (\xi^{i \to 1} + \xi^{1 \to i})$ is an admissible trade achieving the desired rebalancing proportions, $\xi \in \Xi(\nu)$.
\end{proof}

We now have all steps needed to prove existence of optimal rebalancing trade.
\begin{proof}[Proof of Theorem~\ref{thm:OptRebalancing}]
Proposition~\ref{prop:existence} has shown that a feasible rebalancing trade exists, therefore, the optimal rebalancing Problem~\ref{prob:rebalance} is feasible. According to Lemma~\ref{lemma:RebalancingProb},  Problem~\ref{prob:rebalance} is (equivalent to) a linear program and its objective function is bounded from below by zero. For a feasible, bounded linear program an optimal solution exists.  
\end{proof}

\subsection{Relation of rebalancing to non-conservative OT} \label{append:rebalance2OT}
To prove Proposition~\ref{prop:rebalance2OT} we start with a lemma relating a feasible trade to an admissible transport plan.
\begin{lemma}\label{lemma:rebalance2OT}
Let Assumptions A1-A7 hold and denote $v:= \sum_{k=1}^N q_k x_k$.
\begin{enumerate}
\item[a)] If $\xi \in \Xi(\nu)$ is feasible for Problem~\ref{prob:rebalance}, then $\pi^\xi$ defined as
\begin{align*}
\pi^\xi (i,j) = \begin{cases}
    \frac{q_i}{v} \xi_{(i,j)} & \text{ if } i \neq j, \\
    \frac{q_i}{v} \left( x_i - \sum_{k \neq i} \xi_{(i,k)} \right) & \text{ if } i = j,
\end{cases}    
\end{align*}
is feasible for Problem~\ref{prob:Rebalance-OT2}, i.e.~$\pi^\xi \in \Gamma_m (\mu, \nu)$, and it satisfies $\sum_{i,j =1}^N c(i,j) \cdot  \pi^\xi(i,j) < \infty$.

\item[b)] If $\pi \in \Gamma_m (\mu, \nu)$ is feasible for Problem~\ref{prob:Rebalance-OT2} and satisfies $\sum_{i,j =1}^N c(i,j) \cdot  \pi(i,j) < \infty$, then $\xi^\pi$ defined as
\begin{align*}
    \xi^\pi_{(i,j)} = \begin{cases}
        \frac{v}{q_i} \pi(i,j) & \text{ if } i \neq j, \\
        0 & \text{ if } i = j,
    \end{cases}
\end{align*}
is feasible for Problem~\ref{prob:rebalance}, i.e.~$\xi^\pi \in \Xi (\nu)$.
\end{enumerate}
\end{lemma}
\begin{proof}
\begin{enumerate}
\item[a)] $\xi \in \Xi$ being an admissible trade implies $\pi^\xi (i,j) = 0$ for $i \neq j$ with $(i, j) \not\in E$, therefore,  $\sum_{i,j =1}^N c(i,j) \cdot  \pi^\xi(i,j) < \infty$ follows. Using \eqref{eq:Deltax} it can be shown that
\[
\sum_{j=1}^N \pi^\xi (i,j) = \frac{q_i x_i}{v} = \mu_i
\]
and 
\[
\sum\limits_{i=1}^N m(i,j)\pi^\xi (i,j) = \frac{q_j}{v} \left( x_j + \Delta x_j(\xi) \right),
\]
which implies $\pi^\xi \in \Gamma_m (\mu, \nu)$. We leave the details as an exercise. 

\item[b)] Having $\sum_{i,j =1}^N c(i,j) \cdot  \pi(i,j) < \infty$ implies $\pi (i,j) = 0$ for $i \neq j$ with $(i, j) \not\in E$, therefore, $\xi^\pi \in \Xi$ is an admissible trade. Using constraints in $\Gamma_m$ alongside the definitions of $m$ and $\xi^\pi$ it can be shown that
\[
x_i + \Delta x_i (\xi^\pi) = \frac{v}{q_i} \sum_{k=1}^N m(k,i) \pi (k, i),
\]
which implies feasibility $\xi \in \Xi(\nu)$. Again, we leave the details as an exercise. 
\end{enumerate}    
\end{proof}

\begin{proof}[Proof of Proposition~\ref{prop:rebalance2OT}]
Proposition~\ref{prop:existence} in the Appendix~\ref{append:existenceRebalancing} show existence of a feasible trade $\xi \in \Xi (\nu)$. Moreover, objective values of both problems are bounded from below by zero (see Lemma~\ref{lemma:RebalancingProb} in the Appendix). Therefore, optimal values of both problems are finite. From relation~\eqref{eq:costC} it follows
\[
C(\xi) = v \cdot \sum_{i,j = 1}^N c(i,j) \pi^\xi (i,j).
\]
Therefore, the result follows in light of Lemma~\ref{lemma:rebalance2OT}.    
\end{proof}

\bibliographystyle{alpha}
\bibliography{bib.bib}

\end{document}